\documentclass[a4paper,11pt]{article}
\usepackage{amsmath,amsthm}
\usepackage{amssymb}
\usepackage{latexsym}
\usepackage{epsfig}
\usepackage{float}
\usepackage{mathpazo}
\usepackage{fullpage}
\usepackage{enumerate}
\usepackage{paralist}
\usepackage{times}
\usepackage{hyperref}
\usepackage{graphicx}

\newtheorem{theorem}{Theorem}
\newtheorem{corollary}[theorem]{Corollary}

\newtheorem{lemma}[theorem]{Lemma}

\newtheorem{definition}[theorem]{Definition}

\newcommand{\myparagraph}[1]{\paragraph{#1.}}

\newcommand{\bra}[1]{\langle #1|}
\newcommand{\ket}[1]{|#1\rangle}
\newcommand{\braket}[2]{\langle #1|#2\rangle}
\newcommand{\ketbra}[2]{\ket{#1}{\bra{#2}}}

\newcommand{\beq}{\begin{equation}}
\newcommand{\eeq}{\end{equation}}
\newcommand{\trace}{{\rm Tr}}

\newcommand{\norm}[1]{\left\|\,#1\,\right\|}       % norm
\newcommand{\snorm}[1]{\norm{#1}_{\mathrm {\infty}}}    % spectral norm

\newcommand{\set}[1]{{\left\{#1\right\}}}    % braces for set notation
\newcommand{\ve}[1]{\mathbf{#1}}
\newcommand{\abs}[1]{\left\lvert #1 \right\rvert}

\newcommand{\poly}{\operatorname{poly}}

\newcommand{\BPP}{{\rm BPP}}
\newcommand{\BQP}{{\rm BQP}}

\newcommand{\B}{\mathcal{B}}

\newcommand{\complex}{{\mathbb C}}
\newcommand{\reals}{{\mathbb R}}

\newcommand{\nats}{{\mathbb N}}

\newcommand{\comment}[1]{}

\newcommand{\spa}[1]{\mathcal{#1}}

\mathchardef\mhyphen="2D

\newcommand{\ayes}{A_{\rm yes}} %CHECK
\newcommand{\ano}{A_{\rm no}} %CHECK
\newcommand{\nl} {\mathcal{L}_1}
\newcommand{\nll} {\mathcal{L}_2}

\newcommand{\SSC}{{\rm SUCCINCT~SET~COVER}}
\newcommand{\QSSC}{{\rm QUANTUM~SUCCINCT~SET~COVER}}
\newcommand{\IRR}{{\rm IRREDUNDANT}}
\newcommand{\QIRR}{{\rm QUANTUM~IRREDUNDANT}}

\newcommand{\QMMWW}{{\rm QUANTUM~MONOTONE~MINIMUM~WEIGHT~WORD}}
\newcommand{\cqsn}[1]{{\rm cq}\mhyphen\Sigma_{#1}}

\newcommand{\qcsn}[1]{{\rm qc}\mhyphen\Sigma_{#1}}

\newcommand{\cqs}{\cqsn{2}}
\newcommand{\cqslh}[1]{{\rm cq}\mhyphen\Sigma_2{\rm LH}}
\newcommand{\cqslhmin}{{\rm cq}\mhyphen\Sigma_2{\rm LH}\mhyphen{\rm HW}}
\newcommand{\ssat}[1]{{\Sigma_{#1} {\rm SAT}}}

\newcommand{\qcs}{\qcsn{2}}

\newcommand{\ccs}{{\rm cc}\mhyphen\Sigma_2}
\newcommand{\qqs}{{\rm qq}\mhyphen\Sigma_2}
\newcommand{\st}{\s{2}}
\newcommand{\s}[1]{\Sigma_{#1}^p}

\newcommand{\X}{\rm cQMA}
\newcommand{\hin}{H_{\rm in}}
\newcommand{\hprop}{H_{\rm prop}}

\newcommand{\hout}{H_{\rm out}}
\newcommand{\hstab}{H_{\rm stab}}
\newcommand{\hist}{\Pi_{\rm hist}}
\newcommand{\shist}{\spa{S}_{\rm hist}}

\newcommand{\class}[1]{\textup{#1}}
\newcommand{\ppoly}{\class{P}_{\rm/poly}}
\newcommand{\histstate}{\ket{\psi}_{\rm hist}}
\newcommand{\histstatecq}[2]{\ket{#1,#2}_{\rm hist}}
\newcommand{\histstateketbra}{\ketbra{\psi}{\psi}_{\rm hist}}

\newcommand{\histstatecqketbra}[2]{\ketbra{#1,#2}{#1,#2}_{\rm hist}}

\usepackage{color}

\begin{document}

\title{Hardness of approximation for quantum problems}
\author{Sevag Gharibian\footnote{Department of Computer Science, University of Illinois at Chicago, Chicago, IL 60607, USA. Work completed while author was at the David R. Cheriton School of Computer Science and Institute for Quantum Computing, University of Waterloo, Waterloo N2L 3G1, Canada. Supported by the Natural Sciences and Engineering Research Council of Canada (NSERC), NSERC Michael Smith Foreign Study Supplement, EU-Canada Transatlantic Exchange Partnership program, David R.~Cheriton School of Computer Science, and Canadian Institute for Advanced Research.}
\and
 Julia Kempe\footnote{CNRS \& LIAFA,
 University Paris Diderot - Paris 7, Paris, France, and Blavatnik School of Computer Science, Tel Aviv University, Tel Aviv 69978, Israel. Supported
by an Individual Research Grant of the Israeli Science Foundation, by European Research Council (ERC) Starting Grant
QUCO, and by the French ANR Defis program under contract ANR-08-EMER-012 (QRAC project).}}
\date{}
\maketitle

\begin{abstract}
The polynomial hierarchy plays a central role in classical complexity theory. Here, we define a quantum generalization of the polynomial hierarchy, and initiate its study. We show that not only are there natural complete problems for the second level of this quantum hierarchy, but that these problems are in fact hard to approximate. Using these techniques, we also obtain hardness of approximation for the class QCMA. Our approach is based on the use of dispersers, and is inspired by the classical results of Umans regarding hardness of approximation for the second level of the classical polynomial hierarchy [Umans, FOCS 1999]. The problems for which we prove hardness of approximation for include, among others, a quantum version of the Succinct Set Cover problem, and a variant of the local Hamiltonian problem with hybrid classical-quantum ground states.
\end{abstract}

\section{Introduction and Results}\label{scn:intro}

Over the last decades, the Polynomial Hierarchy (PH)~\cite{MS72}, a natural generalization of the class NP, has been the focus of much study in classical computational complexity. Of particular interest is the second level of PH, denoted $\st$. Here, we say a problem is in $\st$ if it has an efficient verifier with the property that for any YES instance $x\in\set{0,1}^n$ of the problem, \emph{there exists} a polynomial length proof $y$ such that \emph{for all} polynomial length proofs $z$, the verifier accepts $x$, $y$ and $z$. Note that the \emph{alternation} from an existential quantifier over $y$ to a for-all quantifier over $z$ is crucial here -- keeping only the existential quantifier reduces us to NP.

It turns out that introducing such {alternating} quantifiers makes $\st$ a powerful class believed to be \emph{beyond} NP. For example, there exist natural and important problems known to be in $\st$ but not in NP. Such problems range from ``does the optimal assignment to a 3SAT instance satisfy \emph{exactly} $k$ clauses?'' to practically relevant problems related to circuit minimization, such as ``given a boolean formula $C$ in Disjunctive Normal Form (DNF), what is the smallest DNF formula $C'$ equivalent to $C$?'' (see, e.g.~\cite{U99}). The study of $\st$ has also led to a host of other fundamental theoretical results, such as the Karp-Lipton theorem, which states that $\class{NP}\not\subseteq \ppoly$ unless PH collapses to $\st$. $\st$ has even been used to prove that SAT cannot be solved simultaneously in linear time and logarithmic space~\cite{F97,FLMV05}. For these reasons, $\st$ and more generally PH have occupied a central role in classical complexity theoretic research.

Moving to the quantum setting, the study of quantum proof systems and a natural quantum generalization of NP, the class Quantum Merlin Arthur (QMA)~\cite{KSV02}, has been a very active area of research over the last decade. Roughly, a problem is in QMA if for any YES instance of the problem, there exists a polynomial size \emph{quantum} proof convincing a {quantum} verifier of this fact with high probability. With the notion of quantum proofs in mind, we thus ask the natural question: \emph{Can a {quantum} generalization of $\st$ be defined, and what types of problems might it contain and characterize?} Perhaps surprisingly, to date there are almost no known results in this direction.

\paragraph{Our results:} In this work, we introduce a quantum generalization of $\st$, which we call $\cqs$, and initiate its study. Our results include $\cqs$-completeness and $\cqs$-hardness of approximation for a number of new problems we define. Our techniques also yield hardness of approximation for the complexity class known as QCMA. We now describe these results in further detail.\\

\myparagraph{Hardness of approximation for $\cqs$}To begin, we informally define $\cqs$ (see Section~\ref{scn:def} for formal definitions).

\begin{definition}[$\cqs$ (informal)]
A problem $\Pi$ is in $\cqs$ if there exists an efficient quantum verifier satisfying the following property for any input $x\in\set{0,1}^n$:
\begin{itemize}
    \item If $x$ is a YES instance of $\Pi$, then {there exists} a \emph{classical} proof $y\in\set{0,1}^{\poly(n)}$ such that {for all} \emph{quantum} proofs $\ket{z}\in\B^{\otimes \poly(n)}$, the verifier accepts $x$, $y$ and $\ket{z}$ with high probability.
    \item If $x$ is a NO instance of $\Pi$, then for all \emph{classical} proofs $y\in\set{0,1}^{\poly(n)}$, there exists a \emph{quantum} proof $\ket{z}\in\B^{\otimes \poly(n)}$ such that the verifier rejects $x$, $y$ and $\ket{z}$ with high probability.
\end{itemize}
\end{definition}

We believe this is a natural quantum generalization of $\st$. Here, the prefix $cq$ in $\cqs$ follows since the existential proof is classical, while the for-all proof is quantum. One can also consider variations of this scheme such as $\qqs$, $\qcs$, or $\ccs$ (with a quantum verifier), defined analogously. In this paper, however, our focus is on $\cqs$, as it is the natural setting for the computational problems for which we wish to prove hardness of approximation. Note also that unlike for $\st$, the definition of $\cqs$ is bounded error -- this is due to the use of a quantum verifier for $\cqs$. This implies, for instance, that the quantum analogue of the classically non-trivial result $\BPP\subseteq\st$~\cite{S83,L83}, i.e.\ $\BQP\subseteq\cqs$, holds trivially. Finally, one can extend the definition of $\cqs$ to an entire hierarchy of quantum classes analogous to PH by adding further levels of alternating quantifiers, attaining presumably different classes depending on whether the quantifier at any particular level runs over classical or quantum proofs.

To next discuss hardness of approximation for $\cqs$, we recall two classical problems  crucial to our work here. First, in the NP-complete problem SET COVER, one is given a set of subsets $\set{S_i}$ whose union covers a ground set $U$, and we are asked for the smallest number of the $S_i$ whose union still covers $U$. If, however, the $S_i$ are represented \emph{succinctly} as the on-set\footnote{By \emph{on-set}, we mean the set of assignments which cause $\phi_i$ to be true.} of a $3$-DNF formula $\phi_i$, we obtain a more difficult problem known as \SSC~(SSC). SSC, along with a related problem \IRR~(IRR), are not just NP-hard, but are $\st$-complete (indeed, they are even $\st$-hard to approximate~\cite{U99}). SSC and IRR are defined as:

\begin{definition}[\SSC~(SSC)~\cite{U99}]
    Given a set $S=\set{\phi_i}$ of $3$-DNF formulae such that $\bigvee_{i\in S}\phi_i$ is a tautology, what is the size of the smallest $S'\subseteq S$ such that $\bigvee_{i\in S'}\phi_i$ a tautology?
\end{definition}

\begin{definition}[\IRR~(IRR)~\cite{U99}]
    Given a DNF formula $\phi=t_1\vee t_2\vee\cdots\vee t_n$, what is the size of the smallest $S\subseteq \set{t_i}_{i=1}^n$ such that $\phi\equiv\bigvee_{i\in S}t_i$?
\end{definition}

Our work introduces and studies quantum generalizations of SSC and IRR. In particular, analogous to the classically important task of circuit minimization, the quantum generalizations we define are arguably natural and related to what one might call ``Hamiltonian minimization'' -- given a sum of Hermitian operators $H=\sum_i H_i$, what is the smallest subset of terms $\set{H_i}$ whose sum approximately preserves certain spectral properties of $H$? We hope that such questions may be useful to physicists in a lab who wish to simulate the simplest Hamiltonian possible while retaining the desired characteristics of a complex Hamiltonian involving many interactions. We remark that at a high level, the connection to $\cqs$ for the task of Hamiltonian minimization is as follows: The classical existential proof encodes the subset of terms $\set{H_i}$, while the quantum for-all proof encodes complex unit vectors which achieve certain energies against $H$. The problem QUANTUM SUCCINCT SET COVER is now defined as follows.

\begin{definition} {\QSSC~(QSSC) (informal)}
    Given a set of local Hamiltonians $\set{H_i}$ such that $\sum_i H_i$ has smallest eigenvalue at least $\alpha$, what is the size of the smallest subset $S$ of the $H_i$ such that $\sum_{H_i\in S} H_i$ has smallest eigenvalue at least $\alpha$? Any subset satisfying this property is called a \emph{cover}.
\end{definition}

Here, a \emph{local Hamiltonian} is a sum of Hermitian operators, each of which acts non-trivially on at most $k\in\Theta(1)$ qubits (hence the name $k$-local Hamiltonian). Intuitively, the goal in QSSC is to cover the entire Hilbert space using as few interaction terms $H_i$ as possible. Hence, we associate the notion of a ``cover'' with obtaining large eigenvalues, as opposed to small ones, making QSSC a direct quantum analogue of SSC. We remark that since SSC is a classical constraint satisfaction problem, we believe the language of \emph{quantum} constraint satisfaction, i.e.\ Hamiltonian constraints, is a natural avenue for defining QSSC. Our first result concerns QSSC, and is as follows.

\begin{theorem}\label{thm:QSSChard}
    QSSC is $\cqs$-complete, and moreover is $\cqs$-hard to approximate within $N^{1-\epsilon}$ for all $\epsilon>0$, where $N$ is the encoding size of the QSSC instance.
\end{theorem}

\noindent By \emph{hard to approximate}, we mean that any problem in $\cqs$ can be reduced to an instance of QSSC via a polynomial time mapping or Karp reduction such that the gap between the sizes of the optimal cover in the YES and NO cases scales as $N^{1-\epsilon}$. In other words, it is $\cqs$-hard to determine whether the smallest cover size of an arbitrary instance of QSSC is at most $g$ or at least $g'$ for $g'/g\in \Omega(N^{1-\epsilon})$ (where $g'\geq g$). We next define the problem QUANTUM IRREDUNDANT~(QIRR).

\begin{definition} {\QIRR~(QIRR) (informal)}
    Given a set of succinctly described orthogonal projection operators $\set{H_i}$ acting on $N$ qubits, and a set $\set{c_i\geq 0}\subseteq \reals$, define $H:=\sum_i c_iH_i$. Then, what is the size of the smallest subset $S\subseteq\set{H_i}$ such that for $H'=\sum_{H_i\in S}c_iH_i$, vectors achieving high and low energies against $H$ continue to obtain high and low energies against $H'$, respectively?
\end{definition}

\noindent Here, by a \emph{succinctly} described projector, we mean a possibly non-local operator which is the tensor product of $k$-local projectors for some $k\in\Theta(1)$. This non-local structure naturally generalizes IRR, where the DNF formula is allowed to be non-local. Our next result is the following.

\begin{theorem}\label{thm:QIRRhard}
    QIRR is $\cqs$-hard to approximate within $N^{\frac{1}{2}-\epsilon}$ for all $\epsilon>0$, where $N$ is the encoding size of the QIRR instance.
\end{theorem}

%We show two further results regarding hardness of approximation. First, via a simple application of the gap amplification technique of Umans~\cite{U99} and the improved disperser construction of Ta-Shma, Umans, and Zuckerman~\cite{TUZ07}, it turns out that the hardness ratios for QSSC and QIRR above can be improved to $N^{1-\epsilon}$ and $N^{\frac{1}{2}-\epsilon}$, respectively (see Corollary~\ref{cor:qsscampgap}).
\myparagraph{Hardness of approximation for QCMA} The techniques from above can be used in a straightforward manner to show hardness of approximation for QCMA. Here, the class QCMA~\cite{AN02} is defined as $\cqs$ with the second (quantum) proof omitted, and can hence be thought of as the first level of our ``$cq$-hierarchy''. By defining the problem QUANTUM MONOTONE MINIMUM SATISFYING ASSIGNMENT (QMSA) (see Section~\ref{scn:QCMA}), we show:

\begin{theorem}\label{thm:QMSAhard}
        QMSA is QCMA-complete, and moreover is QCMA-hard to approximate within $N^{1-\epsilon}$ for all $\epsilon>0$, where $N$ is the encoding size of the QMSA instance.
\end{theorem}

\myparagraph{A canonical $\cqs$-complete problem} Our last results concern a canonical $\st$-complete problem, $\ssat{i}$, and its generalization to the quantum setting. Specifically, given a boolean formula $\phi$, $\ssat{i}$ asks whether:
\[
    \exists \ve{x}_1 \forall \ve{x}_2\exists \ve{x}_3 \cdots\forall \ve{x}_i {\rm ~~such~ that~~ } \phi(\ve{x}_1,\ve{x}_2,\ve{x}_3,\ldots,\ve{x}_i)=1.
\]
Here, we have assumed $i$ is even; for odd $i$, the last quantifier is a $\exists$. The terms $\ve{x}_j$ are vectors of boolean variables. For $i=2$, one can define a natural quantum generalization of this problem, denoted $\cqslh{k}$ and defined in Section~\ref{scn:anotherproblem}, using local Hamiltonians whose ground states are tensor products of a classical string and a quantum state. We show:

\begin{theorem}\label{thm:anothercomplete}
$\cqslh{3}$ is $\cqs$-complete.
\end{theorem}

\noindent Moreover, by defining an appropriate variant of $\cqslh{3}$, denoted $\cqslhmin$ and also defined in Section~\ref{scn:anotherproblem}, where the goal is to minimize the Hamming weight of the classical portion of the ground states mentioned above, we obtain the following result.

\begin{theorem}\label{thm:anothercomplete2}
    $\cqslhmin$ is $\cqs$-complete, and moreover is $\cqs$-hard to approximate within $N^{1-\epsilon}$ for any $\epsilon>0$, for $N$ the encoding size of the $\cqslhmin$ instance.
\end{theorem}

\paragraph{Proof ideas:} Our proofs are inspired by the classical work of Umans~\cite{U99,HSUL02}, and are achieved in a few steps. First, we show a \emph{gap-introducing} reduction from an arbitrary $\cqs$ problem to a problem we call \QMMWW~(QMW) using \emph{dispersers}~(see e.g.,~\cite{SZ94,TUZ07}). We then show the following \emph{gap-preserving} reductions, where $\leq_K$ denotes a mapping or Karp reduction:
\begin{equation}
    \class{QMW} \leq_K \class{QSSC} \leq_K \class{QIRR}\enspace.
\end{equation}
This yields hardness ratios of $N^\epsilon$ for some $\epsilon>0$. To obtain the stronger results claimed in Section~\ref{scn:intro}, we finally apply the gap amplification of Umans~\cite{U99} and improved disperser construction of Ta-Shma, Umans, and Zuckerman~\cite{TUZ07}.

In the classical setting, Umans~\cite{U99,HSUL02} used dispersers to attain hardness of approximation results relative to $\st$ for the classical problems MMWW (the classical version of QMW), SSC and IRR. To extend his techniques to the quantum setting, the most involved aspects of our work are the gap-preserving reductions from QMW to QSSC to QIRR. Here, an intricate balancing act involving carefully defined local Hamiltonian terms is needed to construct operators with the spectral properties required for our reductions. To analyze the resulting sums of non-commuting Hamiltonians, we require heavier machinery, such as the specific structure of Kitaev's local Hamiltonian construction~\cite{KSV02}, the Projection Lemma of Kempe, Kitaev, and Regev~\cite{KKR06}, and the Geometric Lemma of Kitaev~\cite{KSV02}.

Finally, to show $\cqs$-completeness of $\cqslh{3}$, we study the interplay between classical-quantum proofs and Kempe and Regev's~\cite{KR03} $3$-local Hamiltonian construction. Specifically, a careful analysis reveals that any $\cqs$ verification circuit can be modified in such a way that fixing the value $c$ of its classical proof register leads to an \emph{effective} Hamiltonian $H_c$. We then study the spectrum of $H_c$ to achieve the desired result. Moving on to $\cqslhmin$, hardness of approximation is now attained by combining our reduction for $\cqslh{3}$ with the result that QMW is hard to approximate.

\paragraph{Previous and related work:}

In terms of hardness of approximation, the related question of whether a \emph{quantum} PCP theorem holds is currently one of the biggest open problems in quantum complexity theory (see, e.g.,~\cite{A06,AALV09,A10,H12}). Regarding quantum generalizations of PH, the only previous work we are aware of is that of Yamakami~\cite{Y02}. However, the results of Yamakami are largely unrelated to ours (for example, complete problems are not studied), and the proposed definition of Reference~\cite{Y02} differs from ours in a number of ways: It is based on quantum Turing machines (whereas we work with quantum circuits), allows \emph{quantum} inputs (whereas here, like QMA, the input to a problem is a classical string), and considers quantum quantifiers at each level of the hierarchy (whereas in its full generality our scheme allows alternating between classical and quantum quantifiers between levels as desired).

\paragraph{Significance and open questions:}

The classical polynomial hierarchy plays an important role in classical complexity theory, both as a generalization of NP and as a proof tool in itself. It is hoped that the scheme we propose here for generalizing PH to the quantum setting will find similar applications in quantum complexity theory. Second, the problems we show to be $\cqs$-complete here are arguably rather natural, and in embodying a generalization of classical circuit minimization or optimization, may hopefully be related to practical scenarios in a lab. Further, although the alternation between classical and quantum quantifiers in $\cqs$ may \emph{a priori} seem odd, the notion of relating a classical proof to, say, subsets of local Hamiltonian terms, and the quantum proof to quantum states achieving certain energies is in itself quite natural, and in our opinion justifies the study of such a combination of quantifiers. Third, with respect to hardness of approximation, since whether a quantum PCP theorem holds remains a challenging open question, it is all the more interesting that one is able to prove hardness of approximation in a quantum setting here using an entirely different tool, namely that of dispersers. We remark that  dispersers and their two-sided analogues, extractors, have been used classically to amplify existing PCP inapproximability results~\cite{SZ94,Z96}. However, as far as we are aware, neither are known to directly yield PCP constructions.

%Third, our results are the first known hardness of approximation results for a quantum complexity class. Given that whether a quantum PCP theorem holds remains a challenging open question, it is all the more interesting that one is able to prove such results in a quantum setting using an entirely different tool, namely that of dispersers.

We leave a number of questions open: What other natural problems are complete for $\cqs$ or higher levels? Can we say anything non-trivial about the relationship between $\st$ and $\cqs$? How do the different classes $\cqs$, $\qcs$, $\qqs$, and $\ccs$ relate to each other? Where do the quantum hierarchies obtained by extending $\cqs$ to higher levels sit relative to known complexity classes? We hope the answers to such questions will help establish classes like $\cqs$ as fundamental concepts in the setting of quantum computational complexity.

\paragraph{Organization of this paper:}
We begin in Section~\ref{scn:def} by formally defining the classes and problems studied in this paper. In Section~\ref{scn:approxhardness}, we prove that QSSC and QIRR are hard to approximate for $\cqs$ within $N^\epsilon$; this is further improved in Section~\ref{scn:improvements}. Section~\ref{scn:QCMA} presents hardness of approximation results for QCMA. We close in Section~\ref{scn:anotherproblem} by showing $\cqs$-completeness of $\cqslh{3}$ and $\cqs$-hardness of approximation for $\cqslhmin$.

% using simple gap amplification techniques and improved disperser construction

%Intuitively, our definition of $\cqs$~aims to capture the following notion: Given a promise problem $A=\set{\ayes,\ano}$, we say $A\in\cqs$ if there exists a polynomial size quantum circuit $V$ such that whenever $x\in \ayes$, there \emph{exists} a classical string $\ket{\psi_x}_c$ such that \emph{for all} quantum states $\ket{\psi}_q$, $V$ accepts $\ket{\psi_x}_c\otimes\ket{\psi}_q$. To make this rigorous, we use the ``oracle'' definition of $\st$ (see, e.g.~\cite{AB}), which says that $\st=\NP^{\NP}$.

\section{Definitions}\label{scn:def}
We now set our notation, define relevant classes and problems, and state lemmas which prove useful in our analysis.

Beginning with notation, the term $A\succeq B$ means operator $A-B$ is positive semidefinite. The spectral norm of $A$ is $\snorm{A} := \max\{\norm{A\ket{v}}_2 : \norm{\ket{v}}_2 = 1\}$. The projector onto space $\spa{S}$ is $\Pi_{\spa{S}}$. The set of natural numbers is $\nats$. For convenience, we define $\B := \complex^2$, and for a set $S$ of matrices over $\complex$, let $H_S:= \sum_{H_i\in S} H_i$.

We next give a formal definition of $\cqs$. Here, a promise problem is a pair $A=(\ayes,\ano)$ such that $\ayes,\ano\subseteq\set{0,1}^\ast$ and $\ayes \cap \ano = \emptyset$.
\begin{definition}[$\cqs$]\label{def:cqs}
    Let $A=(\ayes,\ano)$ be a promise problem. We say that $A\in\cqs$ if there exist polynomially bounded functions $t,c,q:\nats\mapsto\nats$, and a deterministic Turing machine $M$ acting as follows. For every $n$-bit input $x$, $M$ outputs in time $t(n)$ a description of a quantum circuit $V_x$ such that $V_x$ takes in a $c(n)$-bit proof $\ket{c}$, a $q(n)$-qubit proof $\ket{q}$, and outputs a single qubit. We say $V_x$ \emph{accepts} $\ket{c}\ket{q}$ if measuring its output qubit in the computational basis yields $1$. Then:
    \begin{itemize}
     \item Completeness: If $x\in \ayes$, then $\exists$ $\ket{c}$ such that $\forall$ $\ket{q}$, $V_x$ accepts $\ket{c}\ket{q}$ with probability $\geq2/3$.
     \item Soundness: If $x\in \ano$, then $\forall$ $\ket{c}$, $\exists$ $\ket{q}$ such that ${V_x}$ rejects $\ket{c}\ket{q}$ with probability $\geq2/3$.
    \end{itemize}
\end{definition}

\noindent Note that the completeness and soundness parameters can be amplified to values exponentially close to $1$. Specifically, we use the standard approach of repeating $V_x$ polynomially many times in parallel, except that we only need one copy of the classical register $\mathcal{C}$ for all parallel runs. For any value $c$ placed in $\mathcal{C}$, we think of it as being ``hardwired'' into $V_x$, thus obtaining a quantum verification circuit $V_{x,c}$, which we now apply in parallel to the many copies of the quantum proof $\ket{q}$. The standard weak error reduction analysis for QMA now applies (see, e.g.~\cite{AN02}). Throughout this paper, we refer to this as \emph{error reduction}.

We next define the terms $\X$ circuit, monotone set, QMW, QSSC, and QIRR.

\begin{definition}[$\X$ circuit]\label{def:cQMA}
    Let $n,m\in\nats^+$. A $\X$ circuit $V$ is a quantum circuit receiving $n$ bits in an \emph{INPUT} register and $m$ qubits in a \emph{CHOICE} register, and outputting a single qubit $\ket{a}$. We say:
    \begin{itemize}
        \item $V$ \emph{accepts}  $x\in\set{0,1}^n$ in INPUT if for all $\ket{y}\in\B^{\otimes m}$ in CHOICE, measuring $\ket{a}$ in the computational basis yields $1$ with probability at least $2/3$.
        \item $V$ \emph{rejects}  $x\in\set{0,1}^n$ in INPUT if there exists a $\ket{y}\in\B^{\otimes m}$ in CHOICE such that measuring $\ket{a}$ in the computational basis yields $0$ with probability at least $2/3$.
    \end{itemize}
\end{definition}

\begin{definition}[Monotone set]
    A set $S\subseteq\set{0,1}^n$ is called \emph{monotone} if for any $x\in S$, any string obtained from $x$ by flipping one or more zeroes in $x$ to one is also in $S$.
\end{definition}

\begin{definition}[\QMMWW~(QMW)]\label{def:qmmww}
    Given a $\X$ circuit $V$ accepting exactly a non-empty monotone set $S\subseteq\set{0,1}^n$, and integer thresholds $0\leq g\leq g'\leq n$, output:
    \begin{itemize}
        \item YES if there exists an $x\in\set{0,1}^n$ of Hamming weight at most $g$ accepted by $V$.
        \item NO if all $x\in\set{0,1}^n$ of Hamming weight at most $g'$ are rejected by $V$.
    \end{itemize}
\end{definition}

\noindent Note that clearly ${\rm QMW}\in\cqs$.

\begin{definition}[\QSSC~(QSSC)]\label{def:qssc}
    Let $S:=\set{H_i}$ be a set of $5$-local Hamiltonians $H_i$ acting on $N$ qubits such that $\sum_{H_i\in S} H_i\succeq \alpha I$ for $\alpha>0$. Then, given $\beta\in\reals$ such that $\alpha-\beta \geq 1$ and integer thresholds $0\leq g\leq g'$, output:
\begin{itemize}
    \item YES if there exists $S^\prime\subseteq S$ of cardinality at most $g$ such that $\sum_{H_i\in S^\prime}H_i\succeq\alpha I$.
    \item NO if for all $S^\prime\subseteq S$ of size at most $g'$, $\sum_{H_i\in S^\prime}H_i$ has an eigenvalue at most $\beta$.
\end{itemize}
    Any $S'$ satisfying the YES case is called a \emph{cover}.
\end{definition}

\noindent Note that requiring $\alpha-\beta\in\Omega(1)$ above is without loss of generality, as any instance of QSSC with gap $1/p(N)$ for $p$ a polynomially bounded function can be modified to obtain an equivalent instance with constant gap by multiplying each $H_i$ by $p(N)$~\cite{W09}.

\begin{definition}[\QIRR~(QIRR)]\label{def:qirr}
  Given $S:=\set{c_iH_i}$, where each $H_i$ acts on $N$ qubits and is a tensor product of $5$-local orthogonal projection operators and $c_i\geq 0$ are real. Then, given $\alpha,\beta\in\reals$ such that $\alpha-\beta\geq 1$, and integer thresholds $0\leq g\leq g'$, output:
\begin{itemize}
    \item YES if there exists $S'\subseteq S$ of cardinality at most $g$ such that for all $\ket{\psi}\in\B^{\otimes N}$:
         \begin{itemize}
            \item If $\trace(H_S\ketbra{\psi}{\psi})\geq \alpha$, then $\trace(H_{S'} \ketbra{\psi}{\psi})\geq \alpha$, and
            \item If $\trace(H_S\ketbra{\psi}{\psi})\leq \beta$, then $\trace(H_{S'} \ketbra{\psi}{\psi})\leq \beta$.
         \end{itemize}
    \item NO if for all $S'\subseteq S$ of cardinality at most $g'$, there exists $\ket{\psi}\in\B^{\otimes N}$ with $\trace(H_S\ketbra{\psi}{\psi})\geq \alpha$ and $\trace(H_{S'} \ketbra{\psi}{\psi})\leq \beta$.
\end{itemize}
\end{definition}

\noindent Roughly, QSSC asks how many local interaction terms in a local Hamiltonian one can discard while maintaining the value of the worst assignment. This is intended to mimic the idea of maintaining a tautology for a $3$-DNF formula in SSC classically. Analogous to the relationship between SSC and IRR, QIRR allows possibly non-local Hamiltonian terms so long as they have a succinct description (this generalizes the use of superconstant arity in IRR) and are projectors up to scalar multiplication (this generalizes the requirement that each term $t_i$ in IRR is an AND of variables). QIRR then asks how many interaction terms can be discarded in a sum of such Hamiltonian terms while ensuring that any assignment $\ket{\psi}$ achieves approximately the same value on both the original and modified Hamiltonians.

Next, the key tool enabling the creation of a gap in our reductions is a \emph{disperser} (see e.g.,~\cite{SZ94,TUZ07}).

\begin{definition}[Disperser]\label{def:disperser}
    Let $G=(L,R,E)$ be a bipartite graph with $\abs{L}=2^n$, $\abs{R}=2^m$ and left-degree $2^d$. Then, $G$ is called a \emph{$(k,\epsilon)$-disperser} if, for any subset $L'\subseteq L$ of size $\abs{L'}\geq 2^k$, $L'$ has at least $(1-\epsilon)\abs{R}$ neighbors in $R$. Moreover, if for any pair $(v\in L,i)$, one can compute the $i$th neighbor of $v$ in time polynomial in $n$, then the disperser is called \emph{explicit}.
\end{definition}

Finally, we recall useful known facts from Hamiltonian complexity theory. We first state two lemmas used to bound the eigenvalues of a pair of non-commuting operators.

\begin{lemma}[Kempe, Kitaev, Regev~\cite{KKR06}, Projection Lemma]\label{l:projlemma}
    Let $Y=Y_1+Y_2$ act on Hilbert space $\mathcal{H}=\mathcal{S}+\mathcal{S}^\perp$ for Hamiltonians $Y_1$ and $Y_2$. Denote the zero eigenspace of $Y_2$ as $\mathcal{S}$, and assume the $Y_2$ eigenvectors in $\mathcal{S}^\perp$ have eigenvalue at least $J>2\snorm{Y_1}$. Then, for $\lambda(Y)$ the smallest eigenvalue of $Y$ and $Y|_\spa{S}:=\Pi_{\spa{S}}Y\Pi_{\spa{S}}$,
    \[
        \lambda(Y_1|_{\mathcal{S}})-\frac{\snorm{Y_1}^2}{J-2\snorm{Y_1}}\leq \lambda(Y)\leq\lambda(Y_1|_{\spa{S}})\enspace.
    \]
\end{lemma}

\begin{lemma}[Kitaev, Shen, Vyalyi~\cite{KSV02}, Geometric Lemma, Lemma 14.4]\label{l:pluslemma}
    Let $A_1,A_2\succeq 0$, such that the minimum \emph{non-zero} eigenvalue of both operators is lower bounded by $v$. Assume that the null spaces $\nl$ and $\nll$ of $A_1$ and $A_2$, respectively, have trivial intersection, i.e.\ $\nl\cap\nll=\set{\vec{0}}$. Then
    \begin{equation}
        A_1+A_2 \succeq 2v\sin^2\frac{\alpha(\nl,\nll)}{2}I\enspace,
    \end{equation}
    where the \emph{angle} $\alpha(\spa{X},\spa{Y})$ between $\spa{X}$ and $\spa{Y}$ is defined over unit vectors $\ket{x}$ and $\ket{y}$ as $\cos\left[\angle(\spa{X},\spa{Y})\right] := \max_{\ket{x}\in\spa{X},\ket{y}\in\spa{Y}}\abs{\braket{x}{y}}$.
\end{lemma}

We next recall Kitaev's circuit-to-Hamiltonian construction~\cite{KSV02}. Given a $\cqs$ verification circuit $V=V_L\cdots V_1$ (where without loss of generality, each $V_i$ is a one- or two-qubit unitary) acting on $n$ proof bits (register $A$), $m$ proof qubits (register $B$), and $p$ ancilla qubits (register $C$), this construction outputs a $5$-local Hamiltonian $H$ acting on $A\otimes B\otimes C\otimes D$, where $D$ is a clock register consisting of $L$ qubits.  We then have $H:=\hin+\hout+\hprop+\hstab$, for \emph{penalty} terms as defined below:
\begin{eqnarray*}
    \hin&:=&I_{A,B}\otimes
    \left(\sum_{i=1}^{p} \ketbra{1}{1}_{C_i}\right)\otimes \ketbra{0}{ 0}_D\\
    \hout&:=&I_A\otimes\ketbra{0}{0}_{B_1}\otimes
     I_C\otimes \ketbra{ L}{ L}_D\\
    \hprop &:=& \sum_{j=1}^{L} H_j {\rm,~where~ }H_j{\rm ~is ~defined~ as}\\
    && \hspace{-10mm}-\frac{1}{2}V_j\otimes\ketbra{ j}{ {j-1}}_D -\frac{1}{2}V_j^\dagger\otimes\ketbra{{j-1}}{ j}_D +\frac{1}{2}I\otimes(\ketbra{ j}{ j}+\ketbra{ {j-1}}{ {j-1}})_D\\
    \hstab&:=&I_{A,B,C}\otimes\sum_{i=1}^{L-1}\ketbra{01}{01}_{D_i,D_{i+1}}.
\end{eqnarray*}
Above, the notation $A_i$ refers to the $i$th qubit of register $A$ (similarly for $B$, $C$, $D$). For any prospective proof $\ket{\psi}$ in $\trace(H\ketbra{\psi}{\psi})$, each penalty term has the following effect on the structure of $\ket{\psi}$: $\hin$ ensures that at time zero, the ancilla register is set to zero as it should be for $V$. $\hout$ ensures that at time step $L$ of $V$, measuring the output qubit causes acceptance with high probability. $\hprop$ forces all steps of $V$ appear in superposition in $\ket{\psi}$ with equal weights. Finally, note that for $\hin$, $\hout$, and $\hprop$ above, time $t$ in clock register $D$ is implicitly encoded in unary as $\ket{1^t0^{L-t}}$ (for $\hstab$ above, register $D$ is already explicitly written in unary); $\hstab$ is thus needed to prevent invalid encodings of time steps from appearing in $D$.

We use two important properties of this construction. First, the null space of $\hin+\hprop+\hstab$ is the space of \emph{history states}, which for arbitrary $\ket{\psi}_{A,B}$ are defined as
\begin{equation}\label{eqn:hist}
 \histstate:=\frac{1}{\sqrt{L+1}}\sum_{i=0}^L V_i\cdots V_1 \ket{\psi}_{A,B}\otimes\ket{0}_C\otimes\ket{i}_D.
\end{equation}
For $\cqs$ circuits $V$, it is convenient to define for $c\in\set{0,1}^n$ and $\ket{q}\in\B^m$ the shorthand $\histstatecq{c}{q}:=\histstate$ for $\ket{\psi}=\ket{c}\ket{q}$. The second important property of $H$ we use is that its spectrum is related to $V$ as follows.
\begin{lemma}[Kitaev~\cite{KSV02}]\label{l:kitaev}
    The construction above maps $V$ to $(H,a,b)$ satisfying:
    \begin{itemize}
        \item If there exists a proof $\ket{\psi}$ accepted by $V$ with probability at least $1-\epsilon$, then $\histstate$ achieves $Tr(H\histstateketbra)\leq a$ for $a := \epsilon/(L+1)$.
        \item If $V$ rejects all proofs $\ket{\psi}$, then $H\succeq b I$ for $b\in\Omega\left(\frac{1-\sqrt{\epsilon}}{L^3}\right)$.
    \end{itemize}
\end{lemma}
\section{Hardness of approximation for $\cqs$}\label{scn:approxhardness}

We now show hardness of approximation for $\cqs$ for the problems QMW, QSSC, and QIRR. We begin with a gap-introducing reduction from an arbitary problem in $\cqs$ to QMW. We remind the reader that the hardness ratios obtained here are further strengthened in Section~\ref{scn:improvements}.

\begin{theorem}\label{thm:qmmwwHard}
    There exists a polynomial time reduction which, given an instance of an arbitrary $\cqs$ problem, outputs an instance of QMW with thresholds $g$ and $g'$ satisfying $g'/g\in \Theta(N^\epsilon)$ for some $\epsilon>0$, where $N$ is the encoding size of the QMW instance.
\end{theorem}
\begin{proof}
The reduction follows Theorem 1 of Umans~\cite{U99} closely; the points where we deviate from~\cite{U99} are explicitly noted. Let $\Pi$ be an instance of an arbitrary promise problem $A=(\ayes,\ano)$ in $\cqs$ with encoding size $n$, and whose verification circuit $V$ has a $c(n)$-bit existential proof register and a $q(n)$-qubit for-all proof register. We wish to map $\Pi$ to a \X~circuit $W$ for QMW such that $W$ accepts strings of small or large Hamming weight depending on whether $\Pi\in\ayes$ or $\Pi\in\ano$, respectively. To do so, we follow~\cite{U99} and construct an explicit $(k,1/2)$-disperser $G=(L,R,E)$ with left-degree $2^d$ using Reference~\cite{SZ94}, where $\abs{L}=2^{c(n)+1}$, $\abs{R}=2^{k+d-O(1)}$, and $k:=\gamma\log c(n)$ for $\gamma\in\Theta(1)$ to be set as needed. Note that the value of $d$ depends on the specific disperser construction used --- for the construction of ~\cite{SZ94}, we have $d=4k+O(\log n)$.
Roughly, the idea of Umans is now to have $L$ correspond to assignments for the $c(n)$-bit classical register of $V$, and $R$ to assignments for the classical register of $W$ (in the setting of~\cite{U99}, note that $W$ is a classical circuit). We then \emph{encode} assignments from $L$ by instead choosing neighbor sets in $R$. By exploiting the properties of dispersers, one can ensure that the sizes of the neighbor sets in $R$ chosen vary widely between YES and NO cases for $\Pi$.

Specifically, imagine the vertices in $L$ are arranged into a complete binary tree whose $2^{c(n)}$ leaves denote the $2^{c(n)}$ possible assignments to $V$'s classical register. For convenience, we henceforth use $L$ to mean this tree. Now, let $x\in\set{0,1}^{c(n)}$ denote a leaf of $L$. Then, a subset of vertices $R'\subseteq R$ is said to \emph{encode} $x$ if it contains the union of the neighbor sets of all vertices in the unique path from the root of $L$ to $x$. Figure~\ref{fig:disp} illustrates this encoding scheme. How do the vertices of $R$ then relate to $W$? Each vertex $r\in R$ corresponds to an input bit of $W$ -- setting this $r$th bit to one means we ``choose'' vertex $r$.

\begin{figure}[h]\centering
  \includegraphics[height=5cm]{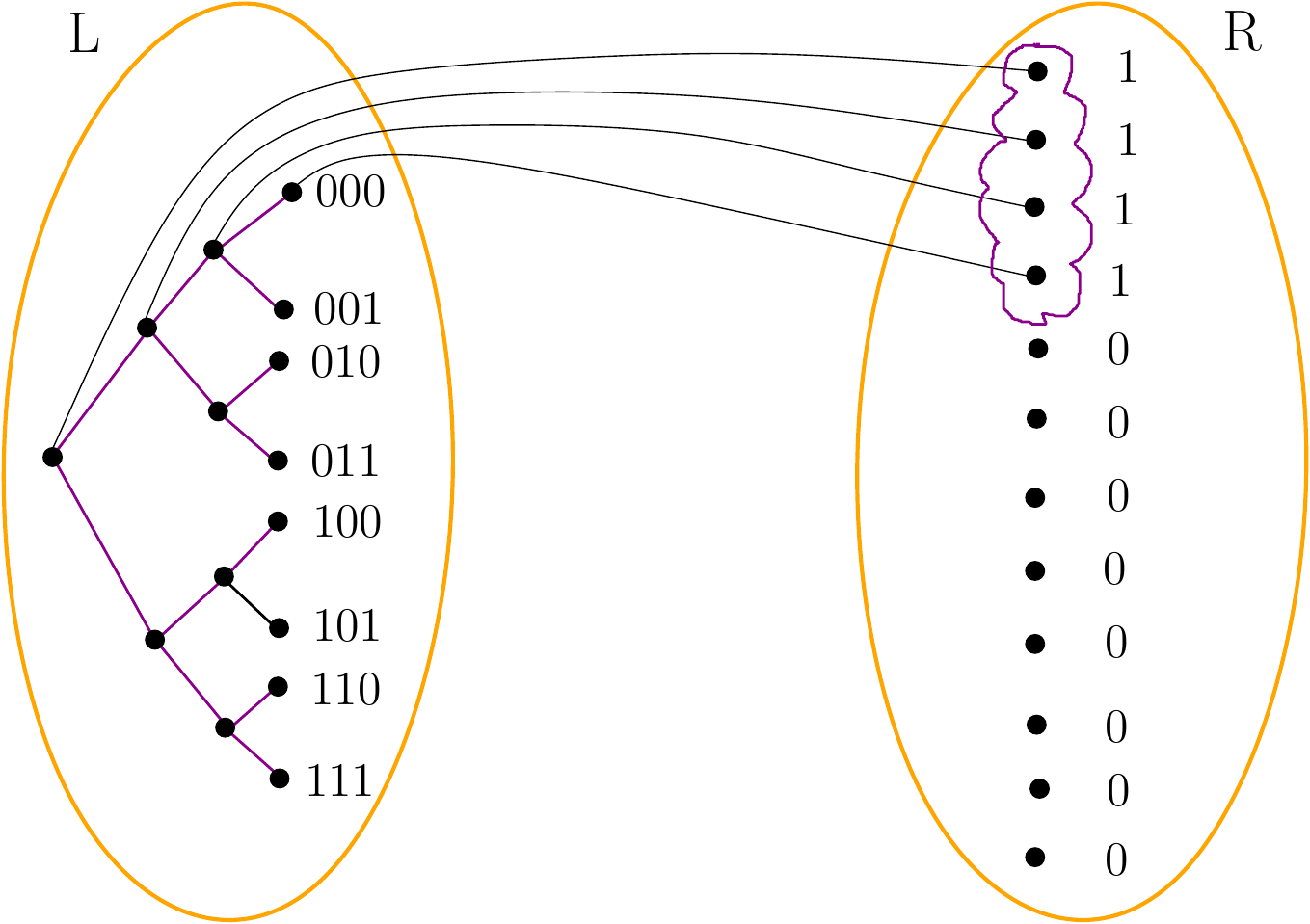}
  \caption{Here, the string $11110\cdots 0$ in $R$ encodes the string $000$ in $L$. (Note: This graph is not a disperser, but nevertheless illustrates the encoding scheme.)}\label{fig:disp}
\end{figure}

With the encoding scheme defined, we now construct the $\X$ circuit $W$. Given $y$ and $\ket{z}$ to its INPUT and CHOICE registers, respectively, $W$ acts as follows: (a) If $y$ corresponds to a subset $R_y\subseteq R$ such that $\abs{R_y}>\abs{R}/2$, then $W$ sets its output qubit to one. (b) If $\abs{R_y}\leq\abs{R}/2$, then $W$ first \emph{decodes} $R_y$ to obtain the set of leaves $L_y\subseteq L$. Roughly, it then outputs one if there exists $x\in L_y$ causing $\Pi$'s verification circuit $V$ to output one when fed the proofs $x$ and $\ket{z}$. These last two steps require further clarification, which we now provide.

First, given $R_y\subseteq R$, decoding it to obtain the set of leaves $L_y\subseteq L$ might \emph{a priori} require exponential time, as recall $\abs{L}=2^{c(n)+1}$. This, however, is precisely where dispersers play their part: Since we set $\epsilon=1/2$ in constructing our disperser, we know that for any $S\subseteq R$ with $\abs{S}\leq \abs{R}/2$, there are at most $2^k=c(n)^\gamma$ vertices in $L$ whose neighbor sets are completely contained in $S$. Thus, by starting at the root of $L$ and performing a breadth-first-search down the tree (where we prune any branches along which we encounter a vertex whose neighbor set is not contained in $R_y$, as by definition such vertices cannot encode any leaf $x$), we can efficiently decode $R_y$ to obtain $L_y$ while visiting only polynomially vertices in $L$. It remains to specify how $W$ checks whether there exists an $x\in L_y$ causing $V$ to accept, and here we must deviate from Umans' construction.

First, if $\abs{L_y}=1$, our task is straightforward -- simply run $V$ as a black box on proofs $x\in L_y$ and $\ket{z}$, and output the result. Then, $W$ outputs one with probability at least $2/3$ on input $y$ for all quantum proofs $\ket{z}$ if and only if $V$ also does so on proofs $x$ and $\ket{z}$. If , however, $\abs{L_y}>1$, a more involved construction of $W$ is necessary. Here, $W$ takes three inputs: a classical description of $V$, an $\abs{R}$-bit string $y$ to denote subsets in $R$, and a $2^kq(n)$-qubit proof $\ket{z}$. Then, for the $i$th candidate string $x_i\in L_y$, $W$ feeds $x_i$ and the $i$th block of $q(n)$ proof qubits of $\ket{z}$ into $V$. (If $\abs{L_y}<2^k$, we simply re-use values of $x\in L_y$ in the leftover parallel runs of $V$.) $W$ then coherently computes the OR of the output qubits of all parallel runs of $V$ and outputs this qubit as its answer.

Let us briefly justify why this works. For simplicity, assume the quantum proof to W can be written $\ket{z}=\ket{z_1}\otimes\cdots\otimes\ket{z_{2^k}}$; entangled proofs can be shown not to pose a problem via the same proof technique used in standard error reduction~\cite{AN02}. Now, if there exists an $x_i\in L_y$ causing $V$ to accept for all quantum proofs, then in the $i$th parallel run of $V$ in $W$ corresponding to $x_i$, $V$ outputs $1$ with probability at least $2/3$ on any $\ket{z_i}$, implying $W$ outputs $1$ with probability at least $2/3$. Conversely, if for all $x_i\in L_y$, there exists a quantum proof $\ket{z_i}$ rejected by $V$, then by standard error reduction for $V$ and the union bound, the state $\ket{z}=\ket{z_1}\otimes\cdots\otimes\ket{z_{2^k}}$ causes $W$ to output $1$ with probability at most $1/3$, as required.

Following Reference~\cite{U99} again, we now argue that $W$ accepts a \emph{non-empty monotone} set, and we analyze the hardness gap introduced by this reduction. The first of these is simple -- namely, $W$ accepts a set $R'\subseteq R$ if either $\abs{R}>\abs{R/2}$, in which case it also accepts any $R''\supseteq R'$, or if $R'$ encodes some $x\in L$ accepted by $V$, in which case any $R''\supseteq R'$ would also encode $x$ and hence be accepted. As for the gap, if $x\in L$ is an accepting assignment for $V$ when $\Pi\in\ayes$, then to encode $x$ using a subset of $R$ requires at most $c(n)2^d$ vertices in $R$, where recall $2^d$ is the left-degree of our disperser. On the other hand, if $\Pi\in\ano$, then the only way for $W$ to accept is to choose $R'\subseteq R$ with $\abs{R'}>\abs{R}/2\approx c(n)^\gamma2^d$. This yields a hardness ratio of $\Omega(c(n)^{\gamma-1})$. Since $W$'s encoding size $N$ is polynomial in $c(n)$, there exists some $\epsilon>0$ such that the ratio produced is of order $N^\epsilon$, as desired.
\end{proof}

We next show a gap-preserving reduction from QMW to QSSC. Its proof requires Lemmas~\ref{l:y2bound} and~\ref{l:ubound}, which are stated and proven subsequently.

\begin{theorem}\label{thm:qmmwwToQssc}
    QSSC is in $\cqs$. Further, there exists a polynomial time reduction which, given an instance of QMW with thresholds $f$ and $f'$, outputs an instance of QSSC with thresholds $g=f+2$ and $g'=f'+2$, respectively.
\end{theorem}
\begin{proof}
That QSSC is in $\cqs$ follows using Kitaev's verifier~\cite{KSV02} for putting $k$-local Hamiltonian in QMA. Specifically, we construct a $\cqs$ verification circuit for QSSC which takes a description $c$ of some subset of local Hamiltonians $S:=\set{H_i}$ in its classical register, and estimates the energy achieved by $\ket{q}$ in its quantum register against $H_S$ using Kitaev's approach, outputting zero or one according to whether the measured energy is above or below the desired thresholds.

To reduce QMW to QSSC, suppose we are given a $\X$ circuit $V$ accepting exactly a non-empty monotone set $T\subseteq\set{0,1}^n$ and threshold parameters $f$ and $f'$. We assume without loss of generality that $V$ is represented as a sequence of one and two qubit unitary gates $V_i$ such that $V=V_L\cdots V_1$. We also assume using standard error reduction that if $V$ accepts (rejects) input $x\in\set{0,1}^n$, then it outputs one (zero) with probability at least $1-\epsilon:=1-2^{-4(n+m)}$.

We now state our instance $(S,\alpha,\beta, g, g')$ of QSSC as follows. We first apply Kitaev's circuit-to-Hamiltonian construction from Section~\ref{scn:def} to $V$ to obtain a $3$-tuple $(H,a,b)$. Note that $H=\sum_{i=1}^r H_i$ with $r$ terms $0\preceq H_i \preceq I$. Then, set $\alpha := 1-(\zeta+1)\epsilon$, and $\zeta:=2(1+2^{2(n+m)})/(L+1)$. Define $\beta:=1-b$. Note that for large $n+m$, this yields $\alpha\geq 1-2^{-(n+m)}$ and $\beta\leq1-c(1-2^{-(n+m)})/{L^3}$ for some constant $c$. Further, define $g:=f+2$, $g':=f'+2$, and let $S$ consist of the elements (intuition to follow)
\begin{eqnarray}
     G_1 &:=& (L+1)\ketbra{0}{0}_{A_1}\otimes I_{B,C}\otimes \ketbra{0}{0}_D \nonumber\\
     &\vdots&\nonumber\\
     G_n &:=& (L+1)\ketbra{0}{0}_{A_n}\otimes I_{B,C}\otimes \ketbra{0}{0}_D\nonumber \\
     G_{n+1} &:=&  (\Delta+1)(\hin+\hprop+\hstab)\nonumber\\
     G_{n+2} &:=& I - (\hin+\hprop+\hstab+\hout),\label{eqn:coverdef}
\end{eqnarray}
for $\Delta\geq 0$ to be chosen as required, and where $A_i$ denotes the $i$th qubit of register $A$. Intuitively, the terms in $S$ play the following roles: $G_{n+1}$ penalizes assignments which are not valid history states. $G_{n+2}$ penalizes valid history states accepted by $V$. Finally, the $G_i$ for $i\in[n]$ penalize valid history states rejected by $V$ (recall that $V$ accepts a monotone set, and so flipping a one to a zero in register $A$ may lead $V$ to reject). Thus, we cover the entire space. We now make this rigorous.

As required by Definition~\ref{def:qssc}, we begin by showing that $S$ itself is a cover, i.e.\ that $G_S\succeq \alpha I_{A,B,C,D}$. First, note that
\begin{equation}
    G_S = I + \sum_{i=1}^n G_i -\hout + \Delta(\hin+\hprop+\hstab).\label{eqn:1}
\end{equation}
It thus suffices to prove that for large enough $\Delta$,
\begin{equation}
    \Delta (\hin+\hprop+\hstab)+\left(\sum_{i=1}^n G_i\right)-\hout\succeq -(\zeta+1)\epsilon I.\label{eqn:keyeqn}
\end{equation}
To show this, we use Lemma~\ref{l:projlemma}, the Projection Lemma, with
\begin{eqnarray}
    Y_1 := \left(\sum_{i=1}^n G_i\right)-\hout,\hspace{8mm}
    Y_2 := \Delta (\hin+\hprop+\hstab).
\end{eqnarray}
Intuitively, the Projection Lemma tells us that by increasing our weight $\Delta$, we can force the smallest eigenvalue of $Y_1+Y_2$ to be approximately the smallest eigenvalue of $Y_1$ restricted to the null space of $Y_2$. In our setting, this implies it suffices to study the smallest eigenvalue of $Y_1$ restricted to the space of all \emph{valid} history states, i.e.\ states of the form of Equation~(\ref{eqn:hist}). Let $\shist$ denote the space of valid history states; note $\shist$ is the null space of $\hin+\hprop+\hstab$. Then, in the notation of Lemma~\ref{l:projlemma}, to lower bound $\lambda(Y_1|_{\shist})$, we invoke Lemma~\ref{l:ubound} to instead upper bound the largest eigenvalue of $(-Y_1)|_{\shist}$. This yields $\lambda(Y_1|_{\shist})\geq -\zeta\epsilon$. Noting that $\snorm{Y_1}\leq n(L+1)+1$, and since by Lemma~\ref{l:y2bound} the smallest non-zero eigenvalue of $Y_2$ scales as $\Omega(\Delta/L^3)$, it follows by Lemma~\ref{l:projlemma} that by setting $\Delta\in\Omega(n^2L^5/\epsilon)$, we have $Y_1+Y_2\succeq -(\zeta+1)\epsilon I$, as desired. This completes the proof that $S$ is a cover.

We now show the desired reduction. Assume first that $V$ accepts a string $x$ of Hamming weight $k$, and let $T\subseteq [n]$ be such that $i\in T$ if and only if $x_i=1$. We claim there exists a cover $S'\subseteq S$ of size $\abs{S'}=k+2$ which consists of $G_{n+1}$, $G_{n+2}$, and the $k$ terms $G_i$ such that $i\in T$. To show this, following the proof above, the analogue of Equation~(\ref{eqn:keyeqn}) which we must prove is
\begin{equation}
    \Delta (\hin+\hprop+\hstab)+\left(\sum_{i\in T} G_i\right)-\hout\succeq -(\zeta+1)\epsilon I.\label{eqn:keyeqn2}
\end{equation}
First, applying Lemma~\ref{l:ubound} again, we lower bound the smallest eigenvalue of $Y'_1:=\left(\sum_{i\in T} G_i\right)-\hout$ restricted to $\shist$ by $-\zeta\epsilon$. Since $\snorm{Y'_1}\leq \snorm{Y_1}$ for $Y_1$ from the previous case of $T=[n]$, the value of $\Delta$ from before still suffices to apply Lemma~\ref{l:projlemma} and conclude that Equation~(\ref{eqn:keyeqn2}) holds, as desired.

Conversely, suppose $V$ rejects any string $x$ of Hamming weight at most $k$. For any $S'\subseteq S$ with $\abs{S'}\leq k+2$, we claim that $G_{S'}$ has an eigenvalue at most $\beta$. To see this, note first that if  $G_{n+2}\not\in S'$, then the state $\histstatecq{1^n}{y}$ attains expected value zero against $G_{S'}$, where note $\beta\geq 0$. Similarly, if $G_{n+1}\not \in S'$, then the state $\ket{1^n}_{A,B,C}\otimes\ket{0}_D$ obtains expected value at most zero against $G$. We conclude that in order to refute the claim that $G$ has an eigenvalue at most $\beta$, we must have $G_{n+1},G_{n+2}\in S'$. This implies that $S'$ contains at most $k$ terms $G_i$ for $i\in [n]$. Then, consider the string $x$ which has ones precisely at these at most $k$ positions $i\in [n]$ corresponding to $G_i\in S'$. It follows that the state $\histstatecq{x}{y}$ lies in the null space of all terms in $S'$ with the possible exception of $G_{n+2}$. Moreover, since $V$ rejects all strings of Hamming weight at most $k$, there exists by the definition of a $\X$ circuit and Lemma~\ref{l:kitaev} a $\ket{y}\in\B^{\otimes m}$ such that
\[
\trace\left(G_{n+2}\histstatecqketbra{x}{y}\right)=1-\trace\left(H\histstatecqketbra{x}{y}\right)\leq 1-b=\beta,
\]
completing the proof.
\end{proof}

%SEV: Up to here, 3-LH version is OK. However, below we lift the ground space of hin + hprop. For 5-LH, this is easy since we know what that ground space looks like - it consists of all valid history states. For 3-LH, however, hprop is no longer PSD, and in fact does NOT map a valid history state to the zero vector (recall it results in a state with invalid counter states).
The following two lemmas are required for the proof of Theorem~\ref{thm:qmmwwToQssc}. Their statements and proofs assume the notation of Theorem~\ref{thm:qmmwwToQssc}.

\begin{lemma}\label{l:y2bound}
    The smallest non-zero eigenvalue of $Y_2 = \Delta (\hin+\hprop+\hstab)$ scales as $\Omega(\Delta/L^3)$.
\end{lemma}
\begin{proof}
    We bound the smallest non-zero eigenvalue of $\hin+\hprop$; it is straightforward to show using the approach of Reference~\cite{KSV02} that the addition of $\hstab$ does not affect this lower bound. Our proof idea here is to ``lift'' the null space of $\hin+\hprop$ so that the smallest non-zero eigenvalue of $\hin+\hprop$ becomes the smallest eigenvalue of the lifted operator, and then apply the Geometric Lemma (Lemma~\ref{l:pluslemma}) to lower bound the latter.

    To begin, recall that the null space of $\hin+\hprop$ consists of all valid history states
    \[
        \histstate=\frac{1}{\sqrt{L+1}}\sum_{i=0}^L V_i\cdots V_1 \ket{\psi}_{A,B}\otimes\ket{0}_C\otimes\ket{i}_D,
    \]
    for any $\ket{\psi}_{A,B}$. (Since we omit $\hstab$ for now, we assume here that the clock register is represented in binary, i.e.\ there are no invalid clock states.) As done in Reference~\cite{KSV02}, our analysis is simplified by first applying the unitary change of basis $W=\sum_{j=0}^L V_1^\dagger\cdots V_j^\dagger\otimes\ketbra{j}{j}$, yielding
    \begin{eqnarray*}
        W\histstate &=& \ket{\psi}_{A,B}\otimes\ket{0}_C\otimes\ket{\gamma}_D\\
        W\hin W^\dagger &=& \hin= I_{A,B}\otimes \left(\sum_{i=1}^p \ketbra{1}{1}_{C_i}\right)\otimes \ketbra{0}{0}_D\\
        W\hprop W^\dagger &=& I_{A,B}\otimes I_C\otimes E_D
    \end{eqnarray*}
    where $\ket{\gamma}:=\left(\frac{1}{\sqrt{L+1}}\sum_{i=0}^L\ket{i}\right)$, and for some operator $E_D$ whose eigenvalues are given by $\lambda_k=1-\cos(\pi k / (L+1))$ for $0\leq k\leq L$ and whose unique zero-eigenvector is $\ket{\gamma}$.

As alluded to above, we now lift the null space of $W(\hin+\hprop)W^\dagger$. Letting $\hist$ denote the projector onto the space of valid history states $\histstate$, this is accomplished by defining
    \begin{eqnarray*}
        A_1 &:=&  W(\hin + p\hist)W^\dagger\\
        A_2 &:=&  W(\hprop + 2\hist)W^\dagger.
    \end{eqnarray*}
Note that $[\hin,\hist]=[\hprop,\hist]=0$, $\snorm{\hin}\leq p$ and $ \snorm{\hprop}\leq 2$. It thus remains to lower bound the smallest eigenvalue of $A_1+A_2$, for which we apply Lemma~\ref{l:pluslemma} to $A_1+A_2$ via the approach of Reference~\cite{KSV02}. For this, we require values for the parameters $v$ and $\alpha(\nl,\nll)$.

For $v$, note that since $A_1$ is a sum of commuting orthogonal projectors, its smallest non-zero eigenvalue is at least $1$ (assuming $p\geq 1$). Similarly, one infers from the spectrum of $E_D$ stated above that the smallest non-zero eigenvalue of $A_2$ scales as $\Omega(1/L^2)$. It follows that $v\in\Omega(1/L^2)$. As for $\alpha(\nl,\nll)$, note that the null spaces $\nl$ and $\nll$ can be written as
    \begin{eqnarray}
        \nl &=& \B^{\otimes(n+m)}_{A,B}\otimes\operatorname{span}(\ket{\psi}~:~\braket{\psi}{0\cdots0}=0)_C\otimes\operatorname{span}(\ket{1},\ldots,\ket{L})_D\oplus\label{eqn:spa1}\\
%        &&\B^{\otimes(n+m)}_{A,B}\otimes\ket{0\cdots 0}_C\otimes\left[\operatorname{span}(\ket{1},\ldots,\ket{L})\cap\operatorname{span}(\ket{\psi}~:~\braket{\psi}{\gamma}=0)\right]_D\label{eqn:spa2}\\
        &&\B^{\otimes(n+m)}_{A,B}\otimes\ket{0\cdots 0}_C\otimes\operatorname{span}(\ket{\psi}~:~\braket{\psi}{\gamma}=0)_D,\label{eqn:spa2}\\
        \nll &=&
        \B^{\otimes(n+m)}_{A,B}\otimes\operatorname{span}(\ket{\psi}~:~\braket{\psi}{0\cdots0}=0)_C\otimes\ket{\gamma}_D.\nonumber
    \end{eqnarray}

\noindent Observe that $\nl \cap \nll = \set{\ve{0}}$, as required by Lemma~\ref{l:pluslemma}. Then, letting $\Pi_{\nl}$ denote the projector onto $\nl$, we analyze
\[
\cos^2\alpha(\nl,\nll)=\max_{\text{unit }\ket{x}\in\nl,\ket{y}\in\nll}\abs{\braket{x}{y}}^2=\max_{\text{unit }\ket{y}\in\nll}\bra{y}\Pi_{\nl}\ket{y}=\max_{\text{unit }\ket{y}\in\nll}\bra{y}\Pi_1+\Pi_2\ket{y},
\]
where $\Pi_1$ and $\Pi_2$ project onto the spaces in Equations~(\ref{eqn:spa1}) and~(\ref{eqn:spa2}), respectively. As $\bra{y}\Pi_2\ket{y}=0$, we simply need to maximize $\bra{y}\Pi_1\ket{y}$, which is equivalent to maximizing $\abs{\braket{\psi}{\gamma'}}^2$ for any unit vector $\ket{\psi}$ in register $D$ and for unnormalized state $\ket{\gamma'}:=(\frac{1}{\sqrt{L+1}}\sum_{i=1}^L\ket{i})$.
%To see this, let $\ket{\phi_1}=\sum_i\alpha_i\ket{a_i}_{A,B,C}\ket{b_i}_D$ and $\ket{\phi_2}=\sum_j\beta_j\ket{a_j}_{A,B,C}\ket{\gamma}_D$ be arbitrary vectors in the spaces $\Pi_1$ and $\nll$, respectively. Then, $\braket{\phi_1}{\phi_2}=\bra{\gamma}_D(\sum_i\alpha_i\beta_i^\ast\ket{b_i}_D)$.
By the Cauchy-Schwarz inequality, this quantity is upper bounded by $L/(L+1)$. We thus obtain the bound $\cos\alpha(\nl,\nll)\leq \sqrt{L/(L+1)}$. Combining this with the identity $2\sin^2\frac{x}{2}=1-\cos x$ and the Maclaurin series expansion for $\sqrt{1+x}$ (where $\abs{x}\leq 1$) yields $2\sin^2\frac{\alpha(\nl,\nll)}{2}\geq\frac{1}{2(L+1)}$. Substituting into Lemma~\ref{l:pluslemma}, the desired result follows.
\end{proof}

\begin{lemma}\label{l:ubound}
    Define $\hist:=\sum_{x\in\set{0,1}^n,y\in\set{0,1}^m}\histstatecqketbra{x}{y}$ as the projector onto $\shist$, let $\zeta:=2(1+2^{2(n+m)})/(L+1)$, and consider $T\subseteq[n]$. Then, if $V$ outputs one with probability at least $1-\epsilon$ for inputs $(x,\ket{y})$ with  $x\in\set{0,1}^n$ such that $x_i=1$ for all $i\in T$ and for all $m$-qubit $\ket{y}$, one has
\[
    \hist\left[\hout-\sum_{i\in T} G_i \right]\hist \preceq \zeta\epsilon I.
\]
\end{lemma}
\begin{proof}
    Define $Z_1:=\hist(-\sum_{i\in T} G_i)\hist$ and $Z_2:=\hist \hout\hist$. Letting $z\in\set{0,1}^n$ denote the characteristic vector of $T$, i.e.\ the $i$th bit of $z$ is set to one if and only if $i\in T$, it follows that any state $\histstatecq{x}{y}$ is an eigenvector of $Z_1$ with eigenvalue $\braket{x}{z}-\abs{T}$. Hence, for example, $
        \trace\left(Z_1\histstatecqketbra{1^n}{y}\right)=0.
    $
    Further, since $V$ accepts a \emph{non-empty} monotone set, it must accept input $(1^n,\ket{y})$ with probability at least $1-\epsilon$, implying
    $
        \trace(Z_2\histstatecqketbra{1^n}{y})\leq \frac{\epsilon}{L+1}.
    $
    This yields an upper bound of
    \[
        \trace((Z_1+Z_2)\histstatecqketbra{1^n}{y})\leq\frac{\epsilon}{L+1}
    \]
    in this simple case. We now show that deviating from $\histstatecq{1^n}{y}$ above cannot increase our expected value against $Z_1+Z_2$ by ``too much''.

    To do so, let $\ket{\phi}=\alpha_1\ket{\phi_1}+\alpha_2\ket{\phi_2}$ be an arbitrary valid history state where $\abs{\alpha_1}^2+\abs{\alpha_2}^2=1$, $\ket{\phi_1}$ is a (normalized) superposition of valid history states where each history state in the superposition has a string $x$ in register $A$ at time zero satisfying $x_i =1$ if $i\in T$, and where $\ket{\phi_2}$ is a valid history state in the space orthogonal to space of all possible states $\ket{\phi_1}$. We thus first have that
    \[
        \trace\left(Z_1\ketbra{\phi}{\phi}\right)\leq0+\alpha_2^2\trace\left(Z_1\ketbra{\phi_2}{\phi_2}\right)\leq \alpha_2^2[(\abs{T}-1)-\abs{T}]\leq-\abs{\alpha_2}^2.
    \]
    %SEV: the terms with phi_1 disappear, and for phi_2, if there is 1 term in phi_2, it must differ in at least 1 position from z, and hence we get a nupper bound on the energy of (|T|-1)-|T| times alpha_2^2. To deal with more than 1 terms in phi_2, note that all cross terms disappear, since the H_i project onto time step 0 and ensure ith bit of both cross terms is 0, and after that we take the inner product of the cross terms, which of course goes to zero unless all bits in A and B registers agree.
    Moving on to $Z_2$, observe that straightforward expansion yields
    \begin{eqnarray*}
        \trace(Z_2\ketbra{\phi}{\phi})&=&\abs{\alpha_1}^2\trace(Z_2\ketbra{\phi_1}{\phi_1})+\abs{\alpha_2}^2\trace(Z_2\ketbra{\phi_2}{\phi_2})\\
        &+& \alpha_1\alpha_2^*\trace(Z_2\ketbra{\phi_1}{\phi_2})+\alpha_1^*\alpha_2\trace(Z_2\ketbra{\phi_2}{\phi_1}).
    \end{eqnarray*}
    To upper bound this quantity, we use the fact that $\braket{a}{b} + \braket{b}{a} \leq \braket{a}{a} + \braket{b}{b}$ for complex vectors $\ket{a}$ and $\ket{b}$. Namely, setting $\ket{a}:=\alpha_1\sqrt{Z_2}\ket{\phi_1}$ and $\ket{b}:=\alpha_2\sqrt{Z_2}\ket{\phi_2}$ yields
    \begin{eqnarray}
        \trace(Z_2\ketbra{\phi}{\phi})&\leq& 2\abs{\alpha_1}^2\trace(Z_2\ketbra{\phi_1}{\phi_1})+2\abs{\alpha_2}^2\trace(Z_2\ketbra{\phi_2}{\phi_2})\nonumber\\
        &\leq& 2\abs{\alpha_1}^2\trace(Z_2\ketbra{\phi_1}{\phi_1})+ 2\abs{\alpha_2}^2\frac{1}{L+1}\label{eqn:last},
    \end{eqnarray}
    where the second inequality follows since $\snorm{Z_2}\leq 1/(L+1)$. Finally, in order to upper bound the term $\trace(Z_2\ketbra{\phi_1}{\phi_1})$ in Equation~(\ref{eqn:last}), observe that since by assumption $\trace(Z_2\histstatecqketbra{x}{y})\leq \frac{\epsilon}{L+1}$ for all $x$ with $x_i=1$ for $i\in T$, and since $\hout$ is a projector, it follows that the norm of $\hout \histstatecq{x}{y}$ is at most $\sqrt{\epsilon/(L+1)}$. Using the Cauchy-Schwarz inequality, this implies that each cross term in the expansion of $\trace(Z_2\ketbra{\phi_1}{\phi_1})$ can contribute a value of magnitude at most $\epsilon/(L+1)$. Since there are at most $2^{2(n+m)}$ such cross terms, and since the non-cross terms are weighted by a convex combination, we hence have the upper bound of $\trace(Z_2\ketbra{\phi_1}{\phi_1})\leq (1+2^{2(n+m)})\epsilon/(L+1)$.
    Combining these bounds, we have
    \begin{eqnarray*}
        \trace((Z_1+Z_2)\ketbra{\phi}{\phi})&\leq& -\abs{\alpha_2}^2+\frac{2\abs{\alpha_1}^2(1+2^{2(n+m)})\epsilon}{L+1}+\frac{2\abs{\alpha_2}^2}{L+1}\\
        &=&\frac{2\abs{\alpha_1}^2(1+2^{2(n+m)})\epsilon+\abs{\alpha_2}^2(1-L)}{L+1}\\
        &\leq&\frac{2(1+2^{2(n+m)})}{L+1}\epsilon\\
        &=&\zeta\epsilon
    \end{eqnarray*}
    where the second inequality holds when $L\geq 1$.
\end{proof}

Finally, we show that QIRR is $\cqs$-hard to approximate.

\begin{theorem}\label{thm:QSSCtoQIRR}
    There exists a polynomial time reduction which, given an instance of an arbitrary $\cqs$ problem $\Pi$, outputs an instance of QIRR with threshold parameters $h$ and $h'$ satisfying $h'/h\in \Theta(N^\epsilon)$ for some $\epsilon>0$, where $N$ is the encoding size of the QIRR instance.
\end{theorem}
\begin{proof}
We begin by applying Theorems~\ref{thm:qmmwwHard} and~\ref{thm:qmmwwToQssc} to reduce the instance of $\Pi$ to an instance $(S=\set{G_i}_{i=1}^{n+2},\alpha,\beta,g,g')$ of QSSC, and henceforth assume the terminology and definitions introduced in Theorem~\ref{thm:qmmwwToQssc}. Recall that any cover in this QSSC instance must include the terms $G_{n+1}$ and $G_{n+2}$. For ease of exposition, we first reduce this instance to QIRR with parameters $h=g+2r-3$ and $h'=g'+2r-3$, where recall $r$ is the number of terms in $H=\sum_{i=1}^r H_i$. This, however, does not suffice to obtain a hardness of approximation gap, as tracing through Theorems~\ref{thm:qmmwwHard} and~\ref{thm:qmmwwToQssc} yields $r\in\omega(g),\omega(g')$, implying $h'/h\rightarrow 1$ as the instance $\Pi$ in Theorem~\ref{thm:qmmwwHard} grows in size. We then slightly modify our reduction to improve the threshold parameters to $h=gr-1$ and $h'=g'r-1$, which yield the desired hardness of approximation gap.

We now state our instance $(T,\gamma,\delta,h,h')$ of QIRR, and follow with an intuitive explanation. For simplicity of exposition, we assume $r$ is a power of two, but our construction can be easily modified to handle the complementary case.
%If $r$ is not a power of 2, simply add up to $r$ terms G_i to the construction below of
%the form (\Delta+1) |1><1|\otimes I\otimes |i><i|. Then, the energy thresholds $\gamma$ and $\delta$ don't change, equation 16 is identical so first direction of proof goes through unchanged, and chaperone qubits mean these extra terms must always be chosen in T', so second direction of proof holds. Finally, the thresholds h and h' change by at most r, but this is inconsequential in terms of the gap.
We also label $H_r = \hout$. We now introduce three registers: a ``tag'' qubit register (denoted $A$), the space the original cover $\mathcal{S}$ acts on (denoted $B$), and $\log r$ ``chaperone'' qubits (denoted $C$). The Hamiltonian terms we define for QIRR, $T:=\set{F_i}_{i=1}^{n+2r-1}$, act on $A\otimes B\otimes C = \B\otimes\B^{\otimes (n+m+p+q)}\otimes\B^{\otimes \log r}$, and are defined as:
\begin{eqnarray*}
    F_1 &:=& \ketbra{0}{0}_A\otimes (G_1)_B\otimes I_C\\
    &\vdots&\\
    F_n &:=& \ketbra{0}{0}_A\otimes (G_n)_B\otimes I_C\\
    F_{n+1} &:=& (\Delta+1)\left[\ketbra{0}{0}_A\otimes (H_1)_B\otimes I_C+\ketbra{1}{1}_A\otimes I_B\otimes\ketbra{0}{0}_C\right]\\
    &\vdots&\\
    F_{n+r-1} &:=&  (\Delta+1)\left[\ketbra{0}{0}_A\otimes (H_{r-1})_B\otimes I_C+\ketbra{1}{1}_A\otimes I_B \otimes\ketbra{r-2}{r-2}_C\right]\\
    F_{n+r} &:=& \ketbra{0}{0}_A\otimes (I-H_1)_B\otimes I_C + \ketbra{1}{1}_A\otimes I_B \otimes \ketbra{r-1}{r-1}_C\nonumber\\
    &\vdots&\\
    F_{n+2r-1} &:=& \ketbra{0}{0}_A\otimes (I-H_r)_B\otimes I_C + \ketbra{1}{1}_A\otimes I_B \otimes\ketbra{r-1}{r-1}_C.\nonumber
\end{eqnarray*}
We set $\gamma:=\alpha+r-1$, $\delta:=\beta+r-1$, $h:=g+2r-3$, and $h':=g'+2r-3$. Note that each $F_j$ is a projection up to scalar multiplication, as required. We now provide the intuition behind the construction. QIRR is stated in terms of projectors $F_j$ (up to scalar multiplication), whereas QSSC is stated in terms of Hermitian operators $G_i$. Hence, in order to move from the latter to the former, a natural idea is to treat each local Hamiltonian term in the sums comprising $G_{n+1}$ and $G_{n+2}$ as distinct terms $F_{n+1},\ldots,F_{n+r-1}$ and $F_{n+r},\ldots,F_{n+2r-1}$, respectively. The problem with this approach is that in order to rigorously argue that the gap between thresholds $g$ and $g'$ for QSSC is preserved when defining thresholds $h$ and $h'$ for QIRR, we would like, for example, that \emph{all} terms $F_j$ making up $G_{n+1}$ are chosen together in any candidate cover $T'\subseteq T$. To address this issue, we introduce the chaperone qubits, which ensure that any candidate $T'$ plays by these rules. In particular, we can make sure that all terms $F_{n+1},\ldots,F_{n+2r-1}$ are chosen in any $T'$, allowing us to rigorously apply our knowledge of the spectra of $G_{n+1}$ and $G_{n+2}$ to the analysis of $F_T$ versus $F_{T'}$.

We now show that if there exists a cover $S^\prime\subseteq S$ for QSSC of size $v$, then there exists a $T'\subseteq T$ such that $\abs{T'}=v+2r-3$ satisfying the conditions for a YES instance of QIRR. Namely, let
\begin{equation}
    T' = \set{F_i}_{i\in [n] \text{ and } G_i\in S'}\cup\set{F_{n+1},\ldots,F_{n+2r-1}}.\label{eqn:succint}
\end{equation}
Note that it suffices to show that $F_{T'}\succeq \gamma I$ (since if $F_{T'}\succeq \gamma I$, then $F_{T}\succeq \gamma I$ as well). To show this, observe first that we can write $F_{T'}=K_1+K_2$, for $K_1$ and $K_2$ defined as:
\begin{eqnarray}
    K_1 &:=& \ketbra{0}{0}_A\otimes\left(\sum_{i\in [n]\text{ and }G_i\in S'} G_i + (\Delta+1)\sum_{i=1}^{r-1}H_i + \sum_{i=1}^{r}(I-H_i)\right)_B\otimes I_C \nonumber \\&=&\ketbra{0}{0}_A\otimes\left(G_{S'}+(r-1)I\right)_B\otimes I_C\label{eqn:K1}\\
    K_2 &:=& \ketbra{1}{1}_A\otimes I_B\otimes \left((\Delta+1)\left(\sum_{i=0}^{r-2}\ketbra{i}{i}\right)+r\ketbra{r-1}{r-1}\right)_C\nonumber\\ &=&\ketbra{1}{1}_A\otimes I_B\otimes\left(rI + (\Delta+1-r)\sum_{i=0}^{r-2}\ketbra{i}{i}\right)_C\label{eqn:K2},
\end{eqnarray}
where we can assume without loss of generality that $\Delta\geq r-1$. Let $\ket{\phi}=a_0\ket{0}_A\ket{\phi_0}_{BC}+a_1\ket{1}_A\ket{\phi_1}_{BC}$ be an arbitrary state acting on this space with $\abs{a_0}^2+\abs{a_1}^2=1$ and for some unit vectors $\ket{\phi_0}_{BC}$ and $\ket{\phi_1}_{BC}$. Then
\begin{eqnarray*}
    \trace(F_{T'}\ketbra{\phi}{\phi})&=&\trace(K_1\ketbra{\phi}{\phi})+\trace(K_2\ketbra{\phi}{\phi})\\
    &=&\abs{a_0}^2\trace(K_1\ketbra{0}{0}\otimes\ketbra{\phi_0}{\phi_0})+\abs{a_1}^2\trace(K_2\ketbra{1}{1}\otimes\ketbra{\phi_1}{\phi_1})\\
    &\geq&\abs{a_0}^2(\alpha + r-1)+\abs{a_1}^2r\\
    &\geq&\gamma,
\end{eqnarray*}
where the first inequality follows since $\trace(X_{AB}I_A\otimes Y_B)=\trace(\trace_A(X_{AB})Y_B)$ and since $G_{S'}$ is a cover by assumption, and the second inequality since $0\leq\alpha\leq 1$. We conclude that $H_{T'}\succeq \gamma I$, as desired.

We now prove the other direction, namely that if there does not exist a cover $S^\prime\subseteq S$ for QSSC of size $v$, then all subsets $T'\subseteq T$ of size $\abs{T'}=v+2r-3$ satisfy the conditions for a NO instance of QIRR. To see this, note first that any candidate $T'$ must include the terms $F_i$ for $n+1\leq i\leq n+r-1$. This is because if, for example, $F_{n+1}\not\in T'$, then vector $\ket{\phi}:=\ket{1}_A\ket{\psi}_B\ket{0}_C$ obtains expected value $\Delta+1\geq \gamma$ against $F_T$, but $\ket{\phi}$ is orthogonal to $F_{T'}$. A similar argument holds for the terms $F_i$ with indices $n+r\leq i\leq n+2r-1$, since state $\ket{\phi}:=\ket{1}_A\ket{\psi}_B\ket{r-1}_C$ obtains expected value $r\geq \gamma$ against $F_T$, but obtains value at most $r-1\leq \delta$ against $F_{T'}$ if there exists an $i\in [n+r,n+2r-1]$ such that $i\not\in T'$. Thus, for any candidate $T'$ of size $v+2r-3$, this leaves $v-2$ terms to be chosen from $\set{F_1,\ldots,F_n}$. If we now restrict ourselves to states of the form $\ket{0}_A\ket{\psi}_{BC}$, we find that we are reduced to the same argument in the NO direction of Theorem~\ref{thm:qmmwwToQssc} -- namely, as $S$ is a cover and any $S'\subseteq S$ of size $v$ is not a cover, there must exist a state $\ket{\phi}:=\ket{0}_A\ket{\psi}_{BC}$ such that
\begin{equation}\label{eqn:qirrsound1}
    \trace(\ketbra{\phi}{\phi}F_T)=\trace\left[\trace_C(\ketbra{\psi}{\psi})(G_S + (r-1)I)\right]\geq \alpha + (r-1)\geq \gamma,
\end{equation}
whereas
\begin{equation}\label{eqn:qirrsound2}
    \trace(\ketbra{\phi}{\phi}F_{T'})=\trace\left[\trace_C(\ketbra{\psi}{\psi})(G_{S'} + (r-1)I)\right]\leq \beta + (r-1)=\delta.
\end{equation}
This concludes the reduction from QSSC to QIRR with parameters $h=g+2r-3$ and $h'=g'+2r-3$.

To obtain improved parameters $h=gr-1$ and $h'=g'r-1$, we modify the construction above as follows (intuition to follow): The terms $F_i$ for $n+1\leq n+2r-1$ from the old construction remain unchanged. For $i\in[n]$, we replace each $F_i := \ketbra{0}{0}_A\otimes (G_i)_B\otimes I_C$ with the $r$ distinct terms:
\begin{eqnarray*}
&F_{i,1}& := \ketbra{0}{0}_A\otimes (G_i)_B\otimes \ketbra{0}{0}_C,\\
    &F_{i,2}& := \ketbra{0}{0}_A\otimes (G_i)_B\otimes \ketbra{1}{1}_C,\\
    &\vdots&\\
    &F_{i,r}& := \ketbra{0}{0}_A\otimes (G_i)_B\otimes \ketbra{r-1}{r-1}_C.
\end{eqnarray*}
Thus, the total number of terms in our QIRR instance increases from $n+2r-1$ to $r(n+2)-1$. Intuitively, we have used the chaperone qubits to split each $F_i$ into $r$ terms $F_{i,j}$, such that if in the old construction we chose $F_i\in T'$, then in the new construction we must place all $r$ terms $F_{i,j}$ in $T'$ in order for the new $F_{T'}$ to maintain its desired spectrum. Thus, whereas the old construction chose $g-2$ terms $F_i$ to place in $T'$, the new construction chooses $r(g-2)$ terms $F_{i,j}$ to place in $T'$, yielding the desired thresholds $h=gr-1$ and $h'=g'r-1$.

The completeness and soundness proofs now follow similarly to the previous case. Namely, given a cover $S'\subseteq S$ for QSSC of size $v$, the set $T'\subseteq T$ with $\abs{T'}=vr-1$ we choose is
\begin{equation}
    T' = \set{F_{i,j}}_{i\in [n] \text{ and } G_i\in S',j \in [r]}\cup\set{F_{n+1},\ldots,F_{n+2r-1}}.
\end{equation}
Since $F_{T'}$ in this new reduction is precisely $F_{T'}$ in the old reduction, the remainder of this direction proceeds identically. Conversely, if there does not exist a cover $S'\subseteq S$ for QSSC of size $v$, we similarly first argue that $F_i$ for $n+1\leq i\leq n + 2r-1$ must be chosen in any candidate $T'\subseteq T$ of size $\abs{T'}=vr-1$, leaving $r(v-2)$ terms to be chosen from $\set{F_{1,1},\ldots, F_{n,r}}$. This implies that for any such $T'$, there must exist a $j\in [r]$ such that the number of terms $F_{i,j}$ in $T'$ is at most $v-2$. Since no cover of size $v$ exists for our QSSC instance, we conclude there exists an appropriate choice of $\ket{\phi}:=\ket{0}_A\ket{\psi}_B\ket{j}_C$ such that Equations~(\ref{eqn:qirrsound1}) and~(\ref{eqn:qirrsound2}) still hold.
\end{proof}

\section{Improvements to hardness gaps}\label{scn:improvements}

We now improve the hardness gaps of Theorems~\ref{thm:qmmwwHard},~\ref{thm:qmmwwToQssc}, and~\ref{thm:QSSCtoQIRR} to obtain the results claimed in Theorems~\ref{thm:QSSChard} and~\ref{thm:QIRRhard}. The key idea is to use the fact that the gap for QMW from Theorem~\ref{thm:qmmwwHard} can be amplified by composing the cQMA circuit $W$ with itself. The results here adapt Section 5 of~\cite{U99} in a simple manner to the quantum setting.

Specifically, assume for the moment that the output qubit of $W$ is actually a classical bit, i.e.\ that the output qubit is given \emph{after} being measured in the computational basis. Then, one can recursively define $W^1:=W$ and $W^t$ as $W^{t-1}$ with $n$ independent copies of $W$ at each of its $n$ INPUT bits. (Note that entanglement between quantum proofs for different copies of $W$ does not affect the soundness of $W^t$, as each $W$ outputs a classical bit, and no quantum proofs are reused.) Now, such a recursive composition of $W$ can easily be made well-defined even if $W$'s output qubit is a superposition of $\ket{0}$ and $\ket{1}$ using the principle of deferred measurement~\cite{NC00} -- namely, without loss of generality, we can assume $W$ first copies its $n$ classical INPUT bits to an ancilla, and henceforth acts only on its CHOICE and ancilla registers. Thus, the output qubit of each copy of $W$ in $W^t$ is effectively used only as a classical control in the remainder of the circuit, and so the measurement of all output qubits can be deferred to the end of $W^t$. Finally, since we can assume using standard error reduction that the completeness and soundness error of $W$ scale as $2^{-n}$, it follows by the union bound that with probability exponentially close to $1$, all the $W$ circuits comprising $W^t$ output the correct answer. In other words, with high probability, one can think of $W^t$ as a composition of zero-error circuits $W$ (where zero-error means zero completeness and soundness error). With this viewpoint, the proof of Lemma 3 of Reference~\cite{U99} directly yields the following result in the quantum setting.

\begin{lemma}\label{l:umansamp}
        If W is a cQMA circuit accepting exactly a monotone set, it follows that:
    \begin{enumerate}
        \item $\abs{W^t}\leq n^t \abs{W}$, where $\abs{W}$ denotes the size of $W$,
        \item $W$ accepts an input of Hamming weight $k$ if and only if $W^t$ accepts an input of weight $k^t$,
        \item $W^t$ accepts exactly a monotone set.
    \end{enumerate}
\end{lemma}

To improve the hardness gap of Theorem~\ref{thm:qmmwwHard}, we now simply replace the cQMA circuit $W$ constructed in the proof of Theorem~\ref{thm:qmmwwHard} with $W^t$ for an appropriate choice of $t$. The details and resulting analysis follow identically to the proof of Theorem 4 of Reference~\cite{U99}, which combined with the improved disperser construction of Reference~\cite{TUZ07} (see Theorem 7.2 therein) yields:

\begin{theorem}\label{thm:qmwampgap}
    QMW is $\cqs$-hard to approximate with gap $N^{1-\epsilon}$ for any $\epsilon>0$, for $N$ the encoding size of the QMW instance.
\end{theorem}

 Using this as the starting point in our reduction chain to QSSC and QIRR, a closer analysis of the proofs of Theorems~\ref{thm:QSSChard} and~\ref{thm:QIRRhard} now yields:
\begin{corollary}\label{cor:qsscampgap}
    QSSC and QIRR are $\cqs$-hard to approximate with gaps $N^{1-\epsilon}$ and $N^{\frac{1}{2}-\epsilon}$ for any $\epsilon>0$, respectively, and where $N$ is the encoding size of the respective QSSC and QIRR instances.
\end{corollary}

%Specifically, defining a \emph{co-ND} circuit as a special case of a cQMA circuit which has zero-error and whose CHOICE register is also classical, we first have in the classical setting:
%
%%\begin{definition}[\cite{U99}]
%%    For C a co-ND circuit, let $C^1:=C$, and recursively define $C^{l}$ as $C^{l-1}$ with $N$ independent copies of $C$ at each of its $n$ INPUT bits.
%%\end{definition}
%\begin{lemma}[\cite{U99}]\label{l:umansamp}
%        For C a co-ND circuit, let $C^1:=C$, and recursively define $C^{t}$ as $C^{t-1}$ with $n$ independent copies of $C$ at each of its $n$ INPUT bits. Then, if $C$ accepts exactly a monotone set, one has:
%    \begin{enumerate}
%        \item $\abs{C^t}\leq n^t \abs{C}$, where $\abs{C}$ denote the size of $C$,
%        \item $C$ accepts an input of Hamming weight $k$ if and only if $C^t$ accepts an input of Hamming weight $k^t$,
%        \item $C^t$ accepts exactly a monotone set.
%    \end{enumerate}
%\end{lemma}
%
%As the proof of Lemma~\ref{l:umansamp} does not  is not difficult to see that the same properties hold for $C^l$ f

\section{Hardness of approximation for QCMA}\label{scn:QCMA}

%~\cite{AN02,JW06,A06,AK07,Beigi08,ABBS08,WY08,JKNN11}

We now briefly remark that the approach of Theorems~\ref{thm:qmmwwHard} and~\ref{thm:qmwampgap} can be adapted to show hardness of approximation for {QCMA}. Our result is a straightforward extension of Umans' classical result~\cite{U99} showing NP-hardness of approximation for the problem MONOTONE MINIMUM SATISFYING ASSIGNMENT.

Specifically, define the problem QUANTUM MONOTONE MINIMUM SATISFYING ASSIGNMENT (QMSA) analogously to QMW, except with the definition of a cQMA circuit $V$ modified to drop the second (quantum) proof, i.e.\ $V$ now only takes one input register comprised of $n$ classical bits. (For example, Definition~\ref{def:cQMA} is modified to say that $V$ \emph{accepts}  $x\in\set{0,1}^n$ in INPUT if measuring $\ket{a}$ in the computational basis yields $1$ with probability at least $2/3$.) Then, it is straightforward to re-run the proofs of Theorems~\ref{thm:qmmwwHard} and~\ref{thm:qmwampgap} without the existence of a second quantum proof register, leading to Theorem~\ref{thm:QMSAhard}.

\section{A canonical $\cqs$-complete problem}\label{scn:anotherproblem}

In this section, we first show that a quantum generalization of the canonical $\st$-complete problem $\ssat{2}$, denoted $\cqslh{k}$, is $\cqs$-complete. We then observe that a similar proof yields $\cqs$-hardness of approximation for an appropriately defined variant of $\cqslh{k}$.

\begin{definition}[$\cqslh{k}$]\label{def:cqslh}
    Given a $3$-local Hamiltonian $H$ acting on $N= n+m$ qubits, and $a,b\in\reals$ such that $a\leq b$ for $b-a \geq 1$, output:
\begin{itemize}
    \item YES if $\exists$ $x\in\set{0,1}^n$ such that $\forall$ $\ket{y}\in\B^{\otimes m}$, $\trace(H \ketbra{x}{x}\otimes\ketbra{y}{y})\geq b$.
    \item NO if $\forall$ $x\in\set{0,1}^n$, $\exists$ $\ket{y}\in\B^{\otimes m}$ such that $\trace(H \ketbra{x}{x}\otimes\ketbra{y}{y})\leq a$.
\end{itemize}
\end{definition}

\begin{theorem}\label{thm:anothercomplete}
$\cqslh{3}$ is $\cqs$-complete.
\end{theorem}
\begin{proof}
That $\cqslh{3}\in\cqs$ follows from Kitaev's verifier for placing $k$-local Hamiltonian in $\class{QMA}$ \cite{KSV02}. As for $\cqs$-hardness, for simplicity we show the result for the case of $\cqslh{5}$ defined with $5$-local Hamiltonians. The proof for the $3$-local case follows identically by instead substituting the $3$-local circuit-to-Hamiltonian construction of Reference~\cite{KR03} below (this is possible because our proof does not exploit the structure of the clock register or $\hstab$).

To see that any instance $\Pi$ of a problem in $\cqs$ reduces to an instance of $\cqslh{5}$, let $V''$ denote the $\cqs$ verification circuit for $\Pi$. Recall that $V''$ acts on a classical proof register $A$, a quantum proof register $B$, and an ancilla register $C$. We begin by modifying $V''$ to obtain a new equivalent circuit $V'$ which first copies the (classical) contents of $A$ to its ancilla register $C$, and henceforth acts on this copied proof in $C$ throughout the verification. This ensures the contents of $A$ remain unchanged during the verification. Next, we modify $V'$ to obtain $V$ by concatenating to its end a Pauli $X$ on the output qubit; this swaps the cases in which $V'$ accepts and rejects, respectively. This is necessary because if $\ket{c}\otimes\ket{q}$ is accepted by $V'$, then $\histstatecq{c}{q}$ obtains low energy against Kitaev's Hamiltonian, whereas in our YES instance here we require high energy. Finally, we apply Kitaev's circuit-to-Hamiltonian construction from Section~\ref{scn:def} on $V$ to obtain a $5$-local Hamiltonian $H$.

Suppose now that we have a YES instance of $\Pi$, i.e.\ there exists bit string $\ket{c}$ such that for all quantum states $\ket{q}$, the circuit $V''$ accepts proof $\ket{c}\otimes\ket{q}$ with probability at least $1-\epsilon$ (and hence $V$ rejects $\ket{c}\otimes\ket{q}$ with probability at least $1-\epsilon$). We show that for all $\ket{\psi}_{B,C,D}$, the state $\ket{c}_A\otimes\ket{\psi}_{B,C,D}$ attains expectation value at least $b$ against $H$, for $b$ from Lemma~\ref{l:kitaev}. In other words, letting $\Pi_c:=(\ketbra{c}{c}_A\otimes I_{B,C,D})$, we claim
\begin{eqnarray*}
    \bra{c}\otimes\bra{\psi}H\ket{c}\otimes\ket{\psi}
    =\bra{c}\otimes\bra{\psi}\Pi_cH\Pi_c\ket{c}\otimes\ket{\psi}
    \geq b.
\end{eqnarray*}
To see this, observe first that
\begin{eqnarray*}
    \Pi_c \hin \Pi_c &=& \ketbra{c}{c}_A\otimes I_B\otimes\left(\sum_{i=1}^{p} \ketbra{1}{1}_{C_i}\right)\otimes \ketbra{0}{0}_D=:\ketbra{c}{c}_A\otimes\hin',\\
    \Pi_c \hout \Pi_c &=& \ketbra{c}{c}_A\otimes\ketbra{0}{0}_{B_1}\otimes
     I_C\otimes \ketbra{L}{L}_D=:\ketbra{c}{c}_A\otimes\hout',\\
    \Pi_c\hstab\Pi_c&=&\ketbra{c}{c}_A\otimes I_{B,C}\otimes\sum_{i=1}^{L-1}\ketbra{01}{01}_{D_i,D_{i+1}}=:\ketbra{c}{c}_A\otimes\hstab'.
\end{eqnarray*}
As for $\Pi_c\hprop\Pi_c$, recall that the verification circuit $V$ consists of two phases: The \emph{copy} phase, consisting of $n$ CNOT gates copying the contents of $A$ to $C$, and the \emph{verification} phase, consisting of the remaining $L-n$ gates of $V$. In other words, we can write
\[
    \hprop=\sum_{j=1}^{n}H_j + \sum_{j=n+1}^L H_j,
\]
where $\sum_{j=1}^{n}H_j$ corresponds to the copy phase and $\sum_{j=n+1}^L H_j$ to the verification phase. Since during the verification phase, $V$ does not act on $A$, we have for all $j> n$ that
\begin{eqnarray*}
    \Pi_c H_j \Pi_c &=& \ketbra{c}{c}_A\otimes\left[-\frac{1}{2}(V_j)_{B,C}\otimes\ketbra{ j}{{j-1}}_D -\frac{1}{2}(V_j^\dagger)_{B,C}\otimes\ketbra{{j-1}}{{j}}_D +\right.\\
    &&\left.\hspace{21mm}\frac{1}{2}I_{B,C}\otimes(\ketbra{{j}}{{j}}+\ketbra{{j-1}}{{j-1}})_D\right]\\
    &=:&\ketbra{c}{c}_A\otimes H_{j}'.
\end{eqnarray*}
As for the copy phase, let $\ketbra{i}{i}\otimes I$ act on $\B\otimes\B$ for $i\in\set{0,1}$. Then, observe that
\[
    (\ketbra{i}{i}\otimes I) \operatorname{CNOT}(\ketbra{i}{i}\otimes I)=\ketbra{i}{i}\otimes X^i,
\]
where $X$ is the Pauli $X$ operator and $X^i=X$ if $i=1$ and $X^i=I$ otherwise. This implies that for any step $j\leq n$, i.e.\ where $V$ applies a CNOT gate with qubit $A_j$ as control and $C_j$ as target, and letting $c_j$ denote the $j$th bit of $c$, we have
\begin{eqnarray*}
    \Pi_c H_{j} \Pi_c &=& \ketbra{c}{c}_A\otimes\left[-\frac{1}{2}X_{C_j}^{c_j}\otimes\ketbra{{j}}{{j-1}}_D -\frac{1}{2}X_{C_j}^{c_j}\otimes\ketbra{{j-1}}{{j}}_D +\right.\\&&\left.\hspace{21mm}\frac{1}{2}I\otimes(\ketbra{{j}}{{j}}+\ketbra{{j-1}}{{j-1}})_D\right]\\
    &=:&\ketbra{c}{c}_A\otimes H_{j}'(c),
\end{eqnarray*}
where the notation $H_{j}'(c)$ means $H_{j}'$ is a function of $c$. Letting $\hprop'(c):=\sum_{i=1}^{n}H'_j+\sum_{i=n+1}^{L}H'_j(c)$ and $H(c):=\hin'+\hout'+\hstab'+\hprop'(c)$, we thus have that $\bra{c}\otimes\bra{\psi} H\ket{c}\otimes\ket{\psi}=\bra{\psi}H(c)\ket{\psi}$. It thus suffices to show that
$
    H(c)\succeq b I.
$

To see this, we return to the circuit $V$, and think of $V$ not as accepting classical input $c$, but rather as corresponding to a set of circuits $\set{V_c}$, where each $V_c$ is just $V$ with $c$ hard-wired into register $A$. In particular, at time step $0\leq j\leq n$, $V_c$ applies $X^{c_j}$ to qubit $C_j$. Taking this interpretation, we observe that for any string $c$, plugging $V_c$ into Kitaev's circuit-to-Hamiltonian yields precisely the Hamiltonian $H(c)$. Thus, since by assumption for our particular choice of $c$, $V''$ accepts $\ket{c}\otimes\ket{q}$ for all quantum proofs $\ket{q}$, it follows that $V_c$ rejects all $\ket{q}$ with probability at least $1-\epsilon$. Hence, Lemma~\ref{l:kitaev} implies $H(c)\succeq b I$, as desired.

The converse direction proceeds similarly. Namely, suppose we have a NO instance of $\Pi$, i.e.\ for all bit strings $\ket{c}$, there exists a quantum proof $\ket{q}$ such that $V''$ rejects $\ket{c}\otimes\ket{q}$ with probability at least $1-\epsilon$. Then, we wish to show that for all $c$, there exists a $\ket{\psi_{c}}$ such that $\bra{\psi_{c}}H(c)\ket{\psi_{c}}\leq a$, for $a$ from Lemma~\ref{l:kitaev}. To show this, fix an arbitrary $c$. Since there exists a $\ket{q}$ such that $V_c$ accepts $\ket{q}$ with probability at least $1-\epsilon$, it follows by Lemma~\ref{l:kitaev} that the history state $\ket{\psi_{c}}:=\sum_{i=0}^LV_i\cdots V_1 \ket{q}_B\otimes \ket{0\cdots 0}_C \otimes \ket{ i}_D$ indeed satisfies
\[
    \bra{\psi_{c}}H(c)\ket{\psi_{c}}=\bra{\psi_{c}}\hin'+\hout'+\hstab'+\hprop'(c)\ket{\psi_{c}}\leq 0+a+0+0= a.
\]
\end{proof}

Note that the proof of Theorem~\ref{thm:anothercomplete} has a special property --- the string $c$ fed into the classical proof register of $V''$ is mapped directly in our reduction to the candidate ground states $\ket{c}\ket{q}$ for $3$-local Hamiltonian $H$. This means, for example, that if there exists a $c$ with the desired properties for a YES instance of our starting $\cqs$ problem, then setting $x=c$ in Definition~\ref{def:cqslh} yields that the $\cqslh{3}$ instance we have mapped to is also a YES instance. It follows that applying the reduction in the proof of Theorem~\ref{thm:anothercomplete} to our hard-to-approximate instance of QMMW from Theorem~\ref{thm:qmwampgap} directly yields Theorem~\ref{thm:anothercomplete2}, i.e.\ that the following variant of $\cqslh{3}$, which we call $\cqslhmin$, is $\cqs$-hard to approximate. Intuitively, $\cqslhmin$ is defined analogously to $\cqslh{3}$, except that here the goal is to minimize the Hamming weight of $x$.

\begin{definition}[$\cqslhmin$]\label{def:cqslh}
    Given a $3$-local Hamiltonian $H$ acting on $N= n+m$ qubits, $a,b\in\reals$ such that $a\leq b$ for $b-a \geq 1$, and integer thresholds $0\leq g\leq g'$, output:
\begin{itemize}
    \item YES if there exists $x\in\set{0,1}^n$ of Hamming weight at most $g$ such that for all $\ket{y}\in\B^{\otimes m}$, $\trace(H \ketbra{x}{x}\otimes\ketbra{y}{y})\geq b$.
    \item NO if for all $x\in\set{0,1}^n$ of Hamming weight at most $g'$, there exists $\ket{y}\in\B^{\otimes m}$ such that $\trace(H \ketbra{x}{x}\otimes\ketbra{y}{y})\leq a$.
\end{itemize}
\end{definition}

\section{Acknowledgements}
We thank Richard Cleve, Ashwin Nayak, Sarvagya Upadhyay, and John Watrous for interesting discussions, and especially Oded Regev for many helpful insights, including the suggestion to think about a quantum version of PH. SG acknowledges support from the NSERC CGS, NSERC CGS-MSFSS, and EU-Canada Transatlantic Exchange Partnership programs, and the David R. Cheriton School of Computer Science at the University of Waterloo. JK is supported by an Individual Research Grant of the Israeli Science Foundation, by European Research Council (ERC) Starting Grant
QUCO and by the Wolfson Family Charitable Trust.

\bibliography{Sevag_Gharibian_Phd_Thesis_Bibliography_Abbrv,approx}

\end{document}